\documentclass[preprint,3p,11pt]{elsarticle}

\journal{Journal of Computational Physics}

\usepackage{tabularx}
\usepackage{lineno}
\usepackage{hyperref}
\usepackage{amsmath}
\usepackage{amsfonts}
\usepackage{amsthm}
\usepackage{wrapfig}%%enable enable 
\usepackage{graphicx}
\usepackage{amssymb}
\usepackage{color}
\usepackage{sidecap}
\usepackage{pdfpages}
\usepackage{physics}
\usepackage{tensor}
\usepackage{float}
\usepackage{yhmath}
\usepackage{stmaryrd}
\usepackage{xfrac}
\usepackage{slashed}

\usepackage[font={footnotesize}]{caption}
\setkeys{Gin}{width=0.9\linewidth,keepaspectratio}

\newcommand{\eqr}[1]{equation #1}
\newcommand{\pfrac}[2]{\frac{\partial #1}{\partial #2}}

\newcommand{\pfraca}[1]{\frac{\partial}{\partial #1}}

\newcommand{\mvec}[1]{\mathbf{#1}}
\newcommand{\gvec}[1]{\boldsymbol{#1}}

\newcommand{\fbar}{\bar{f}}

\newcommand{\gcs}{\nabla_{\mvec{x}}}
\newcommand{\gvs}{\nabla_{\mvec{v}}}

\newcommand{\dtz}{\thinspace dz}

\newtheorem{proposition}{Proposition}

\newtheorem{lemma}{Lemma}
\newtheorem{remark}{Remark}

\newcommand{\cbas}[1]{\gvec{\sigma}_{#1}}

\newcommand{\basis}[1]{\mvec{e}_{#1}}

\newcommand{\nbasis}[1]{\hat{\mvec{e}}_{#1}}

\newcommand{\dbasis}[1]{\mvec{e}^{#1}}
\newcommand{\ndbasis}[1]{\hat{\mvec{e}}^{#1}}

\newcommand{\vnorm}[1]{\lVert{#1}\rVert}

\begin{document}

% Linenumbers
%\linenumbers

\begin{frontmatter}

\title{A Conservative Discontinuous Galerkin Algorithm for Particle Kinetics on Smooth Manifolds}

\author[aff1,aff2]{Grant Johnson\corref{cor1}}
\ead{grj@princeton.edu}

\author[aff2]{Ammar Hakim}

\author[aff2]{James Juno}

\cortext[cor1]{Corresponding author.}

\address[aff1]{Department of Astrophysical Sciences, Princeton University, Princeton, NJ 08544, USA}
\address[aff2]{Princeton Plasma Physics Laboratory, Princeton, NJ 08540, USA}

\begin{abstract}
  A novel, conservative discontinuous Galerkin algorithm is presented
  for particle kinetics on manifolds. The motion of particles on the
  manifold is represented using using both canonical and non-canonical
  Hamiltonian formulations. Our schemes apply to either formulations,
  but the canonical formulation results in a particularly efficient
  scheme that also conserves particle density and energy exactly. The
  collisionless update is coupled to a Bhatnagar-Gross-Krook (BGK)
  collision operator that provides a simplified model for relaxation to
  local thermodynamic equilibrium. An iterative scheme is constructed
  to ensure collisional invariants (density, momentum and energy) are
  preserved numerically. Rotation of the manifold is incorporated by
  modifying the Hamiltonian while ensuring a canonical
  formulation. Several test problems, including a kinetic version of
  the classical Sod-shock problem, Kelvin-Helmholtz instability on the
  surfaces of a sphere and a paraboloid, with and without rotations, is
  presented. A prospectus for further development of this approach to
  simulation of kinetic theory in general relativity is presented.
\end{abstract}

%%Research highlights
% \begin{highlights}
% \item Research highlight 1
% \item Research highlight 2
% \end{highlights}


\end{frontmatter}

%\tableofcontents

\section{Introduction and Background}

The calculus of variations is an extremely powerful tool in classical mechanics. Variations in the action functional yield curves along which particles travel. To specify the action for a system, one introduces a Lagrangian $L(q, \dot{q})$ on the configuration space or a Hamiltonian $H(q, p)$ on the phase space. Applying the action principle — either Euler-Lagrange equations or Hamilton’s equations — yields the equations of motion. The power of this approach is made clear when seeking the equations of motion on a general Riemannian manifold, as required when modeling physical systems in strong gravitational fields. In general relativity (GR) the fundamental object is the spacetime on which the particles move, reacting to gravity \cite{wald2010general, witten2020light}. Gravity no
longer appears as a force; rather, it is a manifestation of spacetime
curvature. On this manifold, point particles, in the collisionless limit, follow geodesics, the
analog of straight lines. The study of kinetic theory in GR, which is
the study of kinetic theory on manifolds, remains a fundamental challenge.

Consider a gas around a compact object like a black hole or a neutron
star: the interaction of this gas is due to the combination of
free-streaming between collisions and external forces (for example, from
electromagnetic fields). In turn, the motions of charged particles
 generate currents that modify the electromagnetic field via
Maxwell equations. Further, the field and matter distribution modify
the gravitational field via Einstein equations. This highly nonlinear
process is poorly understood in the regime of intense
gravity. Presently, this physics is modeled using general relativistic
hydrodynamics (GRHD) \cite{dodin2010vlasov,font1998grhydro,font1999kerr} or magnetohydrodynamics
(GRMHD) equations \cite{anton2006grmhd,komissarov2004general}.  It is difficult to
incorporate kinetic effects into such models. In non-relativistic
settings, one would use closures; however, a naive closure in relativity
leads to problems of causality, restricting the validity of such
closures to a narrow range of parameter space. For example, a
diffusive closure for the heat-flux would mean that heat propagates
instantaneously, violating causality (which states that no causal
phenomena can travel faster than the speed of light) \cite{chandra2015extended, cordeiro2024causality}.

The current state-of-the-art is to directly simulate particle kinetics
around compact objects with a general relativistic particle-in-cell
(GRPIC) method \cite{cerutti2015particle, chen2014electrodynamics,
  crinquand2020multidimensional, galishnikova2023collisionless,
  parfrey2019first, philippov2014ab,galishnikova2025entity}. In GRPIC, Hamilton's equations are used to update the position and momentum of each
particle, and the particles couple to the GR-Maxwell equations via the
electromagnetic fields and the
currents \cite{komissarov2004electrodynamics}. The advantage of GRPIC
is its relative ease of implementation, permitting the inclusion of complex physics in a relatively straightforward manner. However, noise is intrinsic to the PIC method;
hence, for problems in which detailed and delicate distribution
functions are required, one needs to resort to a large number of
particles-per-cell, dramatically increasing costs. Further, the noise
in PIC simulations can act as artificial thermalizing collisions \cite{juno2020noise,nevins2005discrete}, leading to
incorrect growth rates for instabilities and improper saturated
states. In such situations, directly discretizing the Vlasov equations
in 6D phase-space becomes attractive, as the noise-free distribution
function can be obtained.

Continuum methods have not received as much attention, as their
implementations can be complicated and expensive, since a momentum grid
needs to be introduced, resulting in a 6D update. However, high-order
methods, like the discontinuous Galerkin (DG) scheme, allow using
fewer cells in momentum space, thereby reducing the total
degrees-of-freedom (DOFs), while still resulting in accurate
solutions. Furthermore, due to the high local floating-point operation
(flop) cost of DG schemes and the need to only communicate immediate
face neighbor data, one can exploit modern CPU and GPU architectures,
as these perform exceptionally well for schemes that have a high
flop-to-data ratio. The decrease in DOFs, the increased accuracy, and
the availability of modern computing architectures mean that continuum
schemes are now not only feasible but for many problems that require
detailed phase-space structure, they are an ideal approach.

In this paper, we take a step towards developing numerical tools for a
detailed, first-principles understanding of kinetic theory in GR by
designing algorithms first for non-relativistic kinetic theory on
manifolds. Our approach is based on formulating free-streaming
(collisionless) motion using a Hamiltonian formalism and then
constructing a conservative discontinuous Galerkin scheme for the
resulting kinetic equation. We also introduce an approximate collision
operator and construct a conservative, implicit, and asymptotic
preserving iterative scheme. Our schemes are based on our previous
research in developing and extensively applying such schemes for
gyrokinetic and Vlasov equations for use in fusion \cite{ francisquez2020conservative, hakim2020conservative,mandell2020electromagnetic,shi2017gyrokinetic} and space and astrophysical plasma physics
applications \cite{hakim2020alias,juno2018discontinuous}. Here we show the extension of
these schemes to account for motion on curved manifolds, including
the incorporation of effects of rotation. The scheme presented here is specifically
applied to two-dimensional manifolds embedded in flat 
three-dimensional Euclidean space. The extension to an intrinsically
curved four-dimensional manifold requires further mathematical developments, and these
will be presented in a separate publication.

This paper will be laid out as follows: Section \ref{sec:Hamil_systems} introduces Hamiltonian systems and their properties. Section \ref{sec:Hamiltonian_Formulations} then expands upon this by introducing canonical and non-canonical systems. Local equilibrium and moments for these systems and simple collisions are introduced next in Section \ref{sec:Equilibrium Distribution Function}. Section \ref{sec:sphere_theory} extends the theory to include static potentials and background flows, specifically examining the case of a uniformly rotating sphere and the connection to fluid equations. The discontinuous Galerkin discretization scheme for the canonical bracket solver and proofs of conservation and stability properties of the discrete scheme are in Section \ref{sec:scheme}. Section \ref{sec:Benchmarks}  benchmarks the canonical bracket solver for examples in both the highly collisional and collisionless limits across a range of geometries, verifying the proofs from the prior section. Section \ref{sec:norm_momenta} introduces a method that uses normalized momentum basis vectors, leading to a non-canonical bracket. Finally, Section \ref{sec:Conclusions} provides a summary of the key results and directions for future research.
\section{Hamiltonian Systems and Their Properties}
\label{sec:Hamil_systems}

Two major formulations of mechanics are the Lagrangian and Hamiltonian
formulations. The Hamiltonian formulation describes the evolution of a
system through phase space, which consists of
configuration space and momentum space. Phase space is a
$2d$ space in which particle flows are incompressible. The two
pieces that go into a Hamiltonian system are the Poisson bracket operator 
and the Hamiltonian. 

Canonical Hamiltonian systems have been extensively treated in classical texts such as Arnold \cite{arnold1989mathematical} and Goldstein \cite{goldstein2002classical}. As well, generalizations to  non-canonical coordinates are likewise well documented \cite{arnold1989mathematical, cary2009hamiltonian}. The relevant aspects have been summarized here as we later use these properties to analyze the fidelity of the discrete system.

The Poisson bracket is
an antisymmetric, bilinear operator that acts on two functions of
phase space. Let $z^i$, $i=1,\ldots,2d$ be phase space
coordinates. Then, for two arbitrary scalar functions $f(z^i,t)$ and
$h(z^i,t)$, the Poisson bracket is
\begin{align}
  \{f, h \} = \pfrac{f}{z^i} \Pi^{ij} \pfrac{h}{z^j} \label{eq:pb-gen}
\end{align}
where $\Pi^{ij}$ is the antisymmetric Poisson tensor. Furthermore,
the Poisson bracket also satisfies the Jacobi identity
\begin{align}
  \{f, \{g,  h \} \}
  +
  \{g, \{h,  f \} \}
  +
  \{h, \{f,  g \} \} = 0.
\end{align}
One can find special coordinates, called canonical coordinates,
in which the Poisson tensor (and hence the Poisson bracket) takes a
very special form. In these canonical coordinates, $(q^i,p_i)$,
$i=1,\ldots,d$, the Poisson bracket becomes
\begin{align}
  \{f, h \}
  =
  \sum_{i=1}^d
  \left(
  \pfrac{f}{q^i}\pfrac{h}{p_i}
  -
  \pfrac{f}{p_i}\pfrac{h}{q^i}
  \right).
\end{align}
This is called the canonical Poisson bracket operator. In this
case, the Poisson tensor takes the simple form
\begin{align}
  \gvec{\sigma} =
  \left(
   \begin{matrix}
    \mathbf{0} & \mathbf{I} \\
    -\mathbf{I} & \mathbf{0}
   \end{matrix}
  \right)
\end{align}
where $\mvec{I}$ is a $d\times d$ unit matrix. The matrix
$\gvec{\sigma}$ is the fundamental symplectic matrix (playing a role
analogous to the Cartesian metric $\delta_{ij}$ in Riemannian
geometry).

Often, canonical coordinates are not ideally suited to the problem at
hand, and in this case, we need to transform the Poisson tensor to a
new set of coordinates. Consider a general transform,
\begin{align}
  z^i = z^i(q^1,\ldots,q^d, p_1,\ldots,p_d).
\end{align}
Then, the Poisson tensor in this new set of coordinates is given by
\begin{align}
    \Pi^{ij} 
  = \nabla_{\mvec{q}}z^i\cdot\nabla_{\mvec{p}} z^j - \nabla_{\mvec{p}}z^i\cdot\nabla_{\mvec{q}}z^j
\end{align}
for $i=1,\ldots,2d$, where $\nabla_{\mvec{q}}$ and
$\nabla_{\mvec{p}}$ are gradient operators in canonical
coordinates. This general form of the Poisson tensor can be obtained
in several ways, and we will show specific examples later in the paper while
discussing the non-canonical formulation of kinetic theory on curved
spaces.

At this point, the manipulations of the Poisson tensor (and hence the
definition of the Poisson bracket) are purely geometric. To bring 
specific physics into the system, we introduce a scalar function,
$H(z^1,\ldots,z^{2d})$, called the Hamiltonian 
, which determines the equations of motion using the Poisson bracket:
\begin{align}
  \dot{z}^i = \{ z^i, H \} = \Pi^{ij} \pfrac{H}{z^j}. \label{eq:zdot-PB}
\end{align}
These are $2d$ ordinary differential equations (ODEs) that describe
the evolution of particle's phase-space coordinates. If we now introduce
a particle distribution function, $f(z^i,t)$ in phase space, then the
evolution of the distribution is given by
\begin{align}
  \pfrac{f}{t} + \dot{z}^i \pfrac{f}{z^i} = 0.
\end{align}
This equation simply states that the value  of the distribution function remains constant along the characteristics.
By using equations \ref{eq:pb-gen} and \ref{eq:zdot-PB}, this
equation can be rewritten in a generic form
\begin{align}
  \pfrac{f}{t} + \{f, H\} = 0. \label{eq:kin-noncons-2}
\end{align}
The derivation of the kinetic equation from a Poisson bracket and Hamiltonian is a well-established procedure, with a number of studies constructing a diverse array of kinetic systems including the Vlasov-Maxwell system of equations \cite{morrison1980vlasov, marsden1982vlasov} and various reductions of the Vlasov equation, such as gyrokinetics \cite{cary2009hamiltonian, burby2017kinetic}. 

With the kinetic equation defined, we next consider some properties and conserved quantities that arise from this equation with an eye toward constructing numerical schemes that conserve or decay these properties as appropriate. 
The fundamental geometric quantity that appears in this equation is the Poisson tensor $\Pi^{ij}$. 
Everything else is defined in terms of this quantity. 
The phase space Jacobian is given by
\begin{align}
  \mathcal{J} = \frac{1}{\sqrt{\det\gvec{\Pi}}}
\end{align}
where $\gvec{\Pi}$ is the matrix of the components $\Pi^{ij}$. The
characteristic velocities are incompressible in phase space:
\begin{align}
  \pfraca{z^i} (\mathcal{J} \dot{z}^i)
  =
  \pfraca{z^i} \big(\mathcal{J}\Pi^{ij} \pfrac{H}{z^j} \big) = 0. \label{eq:phase_space_incompressibility}
\end{align}
Expanding this
\begin{align}
  \pfraca{z^i} \big(\mathcal{J}\Pi^{ij} \big) \pfrac{H}{z^j}
  +
  \mathcal{J}\Pi^{ij} \frac{\partial^2 H}{\partial z^i \partial z^j}
  =
  0.
\end{align}
The second term vanishes due to the commutativity of the partial derivatives combined with the antisymmetry of $\Pi^{ij}$. We thus obtain the Liouville identities
\begin{align}
  \pfraca{z^i} \big(\mathcal{J}\Pi^{ij} \big) = 0.
\end{align}
From these identities, the antisymmetry of $\Pi^{ij}$, and the commutativity of partial derivatives, we can write the Poisson bracket of two arbitrary functions as a divergence:
\begin{align}
  \{f, h \}
  =
  \frac{1}{\mathcal{J}}
  \pfraca{z^i}
  \left(
  f \mathcal{J} \Pi^{ij} \pfrac{h}{z^j}
  \right)
  =
  \frac{1}{\mathcal{J}}
  \pfraca{z^i}
  \left(
  \pfrac{f}{z^j} \mathcal{J} \Pi^{ij} h
  \right).
  \label{eq:pb-div}
\end{align}

This expression \eqr{\ref{eq:pb-div}} allows for writing the
evolution equation of the distribution function in conservation law
form
\begin{align}
  \pfrac{f}{t} + \frac{1}{\mathcal{J}}\pfraca{z^i} \big( \mathcal{J} \dot{z}^i f
  \big) = 0.
  \label{eq:kin-cons-2}
\end{align}
or
\begin{align}
  \pfrac{f}{t} + \frac{1}{\mathcal{J}}\pfraca{z^i} 
  \big( 
  \mathcal{J}  f \Pi^{ij} \pfrac{H}{z^j}
  \big) = 0.
  \label{eq:kin-cons-3}
\end{align}

From phase space incompressibility,
\eqr{\ref{eq:phase_space_incompressibility}} it also follows that
\begin{align}
  f \{f, H \} &= \{ \frac{f^2}{2}, H \}.
\end{align}
In other words, this quadratic quantity $f \{f, H \}$ can be written as a divergence because the Poisson bracket is
anti-symmetric. We can likewise apply the same identities to $H \{f, H \}$. We thus identify the following conserved quantities upon integration over all of phase space
\begin{align}
  \int \mathcal{J} \{f, H \} \dtz
  =
  \int \mathcal{J} f \{f, H \} \dtz
  =
  \int \mathcal{J} H \{f, H \} \dtz
  = 0,
\end{align}
where we have assumed no contributions from the boundaries in configuration space and that $f$ decays exponentially fast towards the boundaries in momentum space. 
The last two integrals are the quadratic invariants (Casimirs) of the Hamiltonian system.

Using these identities in the kinetic
\eqr{\ref{eq:kin-noncons-2}}, we obtain the conservation laws
\begin{align}
  \int \mathcal{J} \pfrac{f}{t} \dtz &= 0, \\
  \int \mathcal{J} H \pfrac{f}{t} \dtz &= 0, \\
  \int \mathcal{J} f \pfrac{f}{t} \dtz &= 0.
\end{align}
If the phase space transform is time-independent, then
$\mathcal{J}$ is also time-independent. In this case, we can write
these identities as
\begin{align}
  \frac{d}{dt} \int \mathcal{J} f \dtz &= 0, \\
  \int H \pfraca{t}(\mathcal{J} f) \dtz &= 0, \\
  \frac{d}{dt}\int \frac{1}{2} \mathcal{J} f^2 \dtz &= 0.
\end{align}
These represent, respectively, conservation of total number of
particles, particle energy, and the $L_2$-norm of the distribution
function.

In Section\thinspace\ref{sec:scheme}, we construct an energy-conserving discontinuous Galerkin scheme for the kinetic equation.  
Our scheme applies to both the canonical and non-canonical formulations. 
We will show that our scheme conserves the energy exactly even after discretizing time, as the Hamiltonian is time-independent. 
Further, in the case that the Poisson bracket is canonical, we show that using central fluxes conserves the $L_2$-norm, while using upwind fluxes decays the $L_2$-norm monotonically, which is preferable for stability.  
\section{Hamiltonian Formulations of the Collisionless Kinetic Equation on
  Manifolds}\label{sec:Hamiltonian_Formulations}

Consider a $d$-dimensional Riemannian manifold with a local coordinate
basis $\basis{i}$ and a corresponding dual basis $\dbasis{i}$. Let the
covariant components of the metric tensor be
$g_{ij} \equiv \basis{i}\cdot\basis{j}$. In general, these components will
depend on the position on the manifold, $x^k$. We restrict our attention to regions of smooth manifolds away from coordinate singularities so that $g_{ij}$ is non-degenerate, and thus the contravariant metric components $g^{ij}$ are finite everywhere in the domain and the metric determinant is non-zero.  The contravariant metric components, $g^{ij}$, are defined as the components of the inverse metric, satisfying $g^{ij} g_{jk} = \delta\indices{^i_k}$. With these geometric assumptions in place, we can formulate the kinetic theory on such a manifold by first considering free-streaming (collisionless) motion. To do this, we define a Hamiltonian
\begin{align}
  H = \frac{1}{2m} \mvec{p}\cdot\mvec{p} \label{eq:particleHamil}
\end{align}
where $\mvec{p}$ is the particle momentum. We will assume from here on 
the particle mass is $m=1$. The goal is to derive the kinetic equation that
describes the distribution function of particles in the absence of
collisions, adding collisional terms later. We will construct two
Hamiltonian formalisms: canonical and non-canonical. The resulting formalism depends on the choice of momentum space coordinates. 

% We note that the choice to use the particle Hamiltonian equation \ref{eq:particleHamil}, is in contrast to previous Hamiltonian formulations of kinetic equations, such as Morrison \cite{morrison1980vlasov} and Burby \cite{burby2017kinetic}, which formulate Hamiltonian field theories that include the contributions from, e.g., the electromagnetic fields to the Hamiltonian. 
% While we demonstrate in Section~\ref{sec:sphere_theory} the Hamiltonian can also be extended to add additional physics, such as rotation of the manifold, the particular choice being made here is to build our kinetic equation in terms of the motions of the particles separate from external, or self-generated collective, forces. 
% This approach is but one choice, but nevertheless as we will show permits significant flexibility in the formalism, both in our choice of coordinates and choice of physics. 

\subsection{Canonical Formulation}

To construct a canonical Hamiltonian formulation, we will use
covariant momentum components as our momentum-space coordinates; that
is, we write $\mvec{p} = p_i \dbasis{i}$ to obtain the Hamiltonian
\begin{align}
  H = \frac{1}{2} \dbasis{i}\cdot\dbasis{j} p_i p_j =  \frac{1}{2} g^{ij} p_i p_j.
\end{align}
In the $2d$ phase space, we will use the configuration space
coordinates $(x^1,\ldots,x^d)$ and the momentum space coordinates
$(p_1,\ldots,p_d)$. As we show below, this Hamiltonian, combined with the
canonical Poisson bracket operating on arbitrary functions $f$ and $h$ of the phase space,
\begin{align}
  \{f, h\}
  =
  \sum_{i=1}^d
  \left( 
  \pfrac{f}{x^i}\pfrac{h}{p_i} - \pfrac{f}{p_i}\pfrac{h}{x^i}
  \right)
  \label{eq:can-PB}
\end{align}
combine to describe the motion of free-streaming test particles on the manifold. By
free-streaming, we mean particles that follow geodesics on the
manifold, and by test particles, we mean particles that do not alter the
metric itself. We prove this claim below in
Proposition\thinspace\ref{thm:geo-motion}.

Hence, if $f(x^i,p_i,t)$ is the distribution function of particles,
then its evolution is described by the equation
\begin{align}
  \pfrac{f}{t} + \{f,H\} = 0. \label{eq:f-hamil}
\end{align}
In conservation law form, this equation can instead be written as
\begin{align}
  \pfrac{f}{t} + \pfraca{x^k} (\dot{x}^k f) + 
  \pfraca{p_k} (\dot{p}_k f) = 0
  \label{eq:f-cons}
\end{align}
where the characteristic velocities in phase space are given by
\begin{align}
  \dot{x}^k = \{ x^k, H \} = \pfrac{H}{p_k} = g^{ik} p_i
  \label{eq:xdot}
\end{align}
and
\begin{align}
  \dot{p}_k = \{ p_k, H \} = -\pfrac{H}{x^k} 
  = -\frac{1}{2} \pfrac{g^{ij}}{x^k} p_i p_j.
  \label{eq:wdot}
\end{align}

\begin{proposition}\label{thm:geo-motion}
  The equations of motion \eqr{\ref{eq:xdot}} and \eqr{\ref{eq:wdot}}
  represent geodesic motion on the manifold.
\end{proposition}
\begin{proof}
  The goal is to eliminate $p_k$ and their time derivatives and
  derive a second-order ODE for $x^k$. Then, we will show that
  this ODE is identical to the well-known equation for geodesics on
  the manifold. For this, take the time-derivative of
  \eqr{\ref{eq:xdot}} to get
  \begin{align}
    \ddot{x}^k = \pfrac{g^{ik}}{x^m} \dot{x}^m p_i + g^{ik} \dot{p}_i.
  \end{align}
  Now, from \eqr{\ref{eq:xdot}} we have $p_i = g_{in}
  \dot{x}^n$. Hence, using this expression and \eqr{\ref{eq:wdot}} to eliminate
  $p_k$ and their time derivatives, we get
  \begin{align}
    \ddot{x}^k = 
    \left(
    \pfrac{g^{ik}}{x^m} g_{in} 
    -
    \frac{1}{2} g^{ik} \pfrac{g^{st}}{x^i} g_{sm} g_{tn}
    \right)
    \dot{x}^m \dot{x}^n. \label{eqn:geodesic_eqn_in_terms_of_metric_grads}
  \end{align}
  Further, from $g^{ik} g_{in} = \delta\indices{^k_n}$, we must have
  \begin{align}
    \pfraca{x^m} (g^{ik} g_{in}) = 0
  \end{align}
  and hence we get the ``shift-identity''
  \begin{align}
    \pfrac{g^{ik}}{x^m} g_{in} = -g^{ik} \pfrac{g_{in}}{x^m}.
  \end{align}
  Using this shift-identity repeatedly, we get
  \begin{align}
    \ddot{x}^k 
    &= 
    -g^{ik} \left(
    \pfrac{g_{in}}{x^m}
    +
    \frac{1}{2} \pfrac{g^{st}}{x^i} g_{sm} g_{tn}
    \right)
    \dot{x}^m \dot{x}^n \notag \\
    &= 
    -g^{ik} \left(
    \pfrac{g_{in}}{x^m}
    -
    \frac{1}{2} g^{st}\pfrac{g_{sm}}{x^i} g_{tn}
    \right)
    \dot{x}^m \dot{x}^n \notag \\
    &= 
    -g^{ik} \left(
    \pfrac{g_{in}}{x^m}
    -
    \frac{1}{2} \pfrac{g_{nm}}{x^i}
    \right)
    \dot{x}^m \dot{x}^n.
  \end{align}
  Now, note that we have
  \begin{align}
    \pfrac{g_{in}}{x^m} \dot{x}^m \dot{x}^n =  \pfrac{g_{im}}{x^n} \dot{x}^m \dot{x}^n
  \end{align}
  as $m$ and $n$ are dummy indices. Hence, we can write
  \begin{align}
    \ddot{x}^k 
    =
    -\frac{1}{2} g^{ik} \left(
    \pfrac{g_{in}}{x^m} + \pfrac{g_{im}}{x^n}
    -
    \pfrac{g_{nm}}{x^i}
    \right)
    \dot{x}^m \dot{x}^n.
  \end{align}
  Utilizing the definition of the Christoffel symbols
  \begin{align}
    \Gamma\indices{^k_{m n}}
    =
    \frac{1}{2} g^{ik} \left(
    \pfrac{g_{in}}{x^m} + \pfrac{g_{im}}{x^n}
    -
    \pfrac{g_{nm}}{x^i}
    \right),
  \end{align}
  we are finally led to the second-order ODE that describes the
  particle motion:
  \begin{align}
    \ddot{x}^k  + \Gamma\indices{^k_{m n}} \dot{x}^m \dot{x}^n = 0, \label{eqn:geodesic}
  \end{align}
  which we can recognize as the equation for a geodesic on a manifold. 
\end{proof}

This proposition shows that all we need is the canonical Poisson
bracket and the corresponding Hamiltonian to update
the distribution function of free-streaming test-particles on a
manifold. We emphasize that the Christoffel symbols never explicitly appear. Further, the Jacobian of the transform, independent of the
coordinate system we choose, is always just unity because the determinant of the canonical Poisson tensor is 1. We show specific examples of how these geodesics evolve on the surface of 
a sphere and in hyperbolic geometry in \ref{app:geodesics}.  

We also remark that although the formalism does not depend on the dimension
or signature of the manifold, in relativistic physics it requires
modification, as the time coordinate needs to be treated on the same
footing as the spatial coordinates, but one must also account for the Lorentzian signature of spacetime.

\subsection{Non-Canonical Formulation}

Though we use a canonical Hamiltonian formulation to construct our
solver, we now show how one can construct a non-canonical formulation
and derive the corresponding Poisson bracket. To start, one constructs
the phase space Lagrangian
\begin{align}
  \mathcal{L} &= p_i \dot{x}^i - H.
\end{align}
As is well known, the Euler-Lagrange equations for this Lagrangian
yield the Hamiltonian equations of motion. Coordinate transforms can
now be applied to this phase space Lagrangian. As an example, we will
use the contravariant components of the momentum, $p^i$,
instead of the covariant components. The transform is simply
\begin{align}
  p_i = g_{ij} p^j.\label{eq:p_i_trans}
\end{align}
With this, the phase space Lagrangian becomes
\begin{align}
  \mathcal{L} &= g_{ij}\dot{x}^i p^j - H \label{eq:pp-lag-H}
\end{align}
with the Hamiltonian now given by
\begin{align}
  H = \frac{1}{2} g_{ij} p^i p^j.
\end{align}
One can substitute this into the phase space Lagrangian and use the
Euler-Lagrange equations to derive the equations of motion. However,
we wish to derive the Poisson bracket for this choice of
coordinates. Let $z^i = (x^1,\ldots,x^d,p^1,\ldots,p^d)$ be phase space
coordinates. Then, the equations of motion can be written in the form
\begin{align}
  \dot{z}^i = \{ z^i, H \} = \Pi^{ij} \pfrac{H}{z^j} \label{eq:zdot-Pi-H}
\end{align}
where now $\{f, h \}$ is the Poisson bracket corresponding to the new
set of coordinates, given by
\begin{align}
  \{f, h \} = \pfrac{f}{z^i} \Pi^{ij} \pfrac{h}{z^j} \label{eq:general_non_can_bracket}
\end{align}
and $i,j = 1,\ldots,2d$. (The index ranges will not usually be
specified explicitly and should be clear from the context).

Hence, to determine the Poisson bracket (or, equivalently,
the Poisson tensor $\Pi^{ij}$), we will work directly with
\eqr{\ref{eq:pp-lag-H}}, deriving the equations of motion in the form
\eqr{\ref{eq:zdot-Pi-H}} and reading off the Poisson tensor.

The Euler-Lagrange equation
\begin{align}
  \frac{d}{dt}\left( \pfrac{\mathcal{L}}{\dot{p}^k} \right) = \pfrac{\mathcal{L}}{p^k}
\end{align}
yields
\begin{align}
  0 = g_{ik}\dot{x}^i - \pfrac{H}{p^k}
\end{align}
or (index relabeling)
\begin{align}
  \dot{x}^i = g^{im}\pfrac{H}{p^m}. \label{eq:sol-xi-nc}
\end{align}

The Euler-Lagrange equation
\begin{align}
  \frac{d}{dt}\left( \pfrac{\mathcal{L}}{\dot{x}^k} \right) = \pfrac{\mathcal{L}}{x^k}
\end{align}
yields
\begin{align}
  \pfrac{g_{kj}}{x^m} \dot{x}^m p^j + g_{kj}\dot{p}^j
  = 
  \pfrac{g_{ij}}{x^k} \dot{x}^ip^j - \pfrac{H}{x^k}.
\end{align}
Rearranging this, we get
\begin{align}
  \dot{p}^s
  =
  g^{ks}\left(
  \pfrac{g_{ij}}{x^k}
  -
  \pfrac{g_{kj}}{x^i}
  \right)
  \dot{x}^ip^j
  - g^{ks} \pfrac{H}{x^k}.
\end{align}
Substituting \eqr{\ref{eq:sol-xi-nc}} into this we get (index relabeling)
\begin{align}
  \dot{p}^s
  =
  \Omega^{sm}
  \pfrac{H}{p^m}
  - g^{sm} \pfrac{H}{x^m}. \label{eq:sol-wi-nc}
\end{align}
where
\begin{align}
  \Omega^{sm}
  \equiv
    g^{ks}\left(
  \pfrac{g_{ij}}{x^k}
  -
  \pfrac{g_{kj}}{x^i}
  \right)
  g^{im} p^j\label{eqn:Omega}
\end{align}
is an anti-symmetric tensor. Hence, using \eqr{\ref{eq:sol-xi-nc}} and
\eqr{\ref{eq:sol-wi-nc}}, we can read off the Poisson tensor as
\begin{align}
  \gvec{\Pi} 
  =
  \begin{bmatrix}
    0 && g^{sm} \\
    -g^{sm} && \Omega^{sm}
  \end{bmatrix}.
\end{align}
From this, we can compute the Poisson bracket for the choice of momentum coordinates in equation \ref{eq:p_i_trans}
\begin{align}
  \{f, h \}  
  =
  g^{sm}
  \left(
  \pfrac{f}{x^s}\pfrac{h}{p^m}
  -
  \pfrac{f}{p^s}\pfrac{h}{x^m}
  \right)
  +
  \Omega^{sm} \pfrac{f}{p^s}\pfrac{h}{p^m}
  \label{eq:non-can-PB}
\end{align}
Note the complexity of this non-canonical Poisson bracket: the metric
and its derivatives explicitly appear, in contrast to the canonical
Poisson bracket\footnote{There is significant structural similarity
  between the non-canonical Poisson Bracket \eqr{\ref{eq:non-can-PB}}
  and the Poisson bracket for charged particle motion in a given
  magnetic field. That Poisson bracket is
  \begin{align*}
    \{f,h\} = 
    \frac{1}{m}
    \big(
    \gcs f \cdot \gvs h - \gvs f \cdot \gcs h
    \big)
    +
    \frac{q\mvec{B}}{m^2} \cdot 
    \gvs f \times \gvs h \label{eq:vm-bracket}    
  \end{align*}
  In fact, this PB is identical if we identify
  $\Omega^{mn} = \epsilon^{imn} B_i/\sqrt{g}$ (ignoring mass and
  charge factors), where $\epsilon^{imn}$ is the completely
  anti-symmetric Levi-Civita symbol.}.

A straightforward calculation shows that the non-canonical Poisson
bracket (or the phase space Lagrangian with the Hamiltonian expression
substituted into it) gives the expected equations of motion
(characteristic velocities in phase space)
\begin{align}
  \dot{x}^s &= p^s \\
  \dot{p}^s &= -\Gamma\indices{^s_{ij}} p^i p^j.
\end{align}
In terms of these characteristic velocities, the non-conservative form
of the kinetic equation is
\begin{align}
  \pfrac{\fbar}{t}
  +
  \dot{x}^s\pfrac{\fbar}{x^s}
  +
  \dot{p}^s
  \pfrac{\fbar}{p^s}
  =
  0
\end{align}
or
\begin{align}
  \pfrac{\fbar}{t}
  +
  p^s\pfrac{\fbar}{x^s}
  -
  \Gamma\indices{^s_{ij}} p^i p^j
  \pfrac{\fbar}{p^s}
  =
  0.
  \label{eq:non-can-nc-kin}  
\end{align}
Here, $\fbar(x^i,p^i,t)$ is the distribution function, now depending
on the contravariant components of the momentum.

The Jacobian of the non-canonical Hamiltonian system is
$\mathcal{J} = 1/\sqrt{\det \gvec{\Pi}} = g$, where $g$ is the
determinant of the matrix of covariant components of the metric
tensor. In terms of this Jacobian, phase space
incompressibility is
\begin{align}
  \pfraca{x^s}( g\dot{x}^s )
  +
  \pfraca{p^s}( g\dot{p}^s )
  =
  0.
\end{align}
To prove this relation we need to use the identity
\begin{align}
  \Gamma\indices{^i_{ij}}
  =
  \frac{1}{\sqrt{g}}\pfrac{\sqrt{g}}{x^j}.
\end{align}
Finally, the phase space incompressibility allows us to write the
conservation law form of the kinetic equation in the non-canonical
coordinates as
\begin{align}
  \pfrac{\fbar}{t}
  +
  \frac{1}{g}\pfraca{x^s}(g \dot{x}^s \fbar)
  +
  \pfraca{p^s}(\dot{p}^s \fbar)
  =
  0.
  \label{eq:non-can-kin}
\end{align}

\section{Equilibrium Distribution Function}\label{sec:Equilibrium Distribution Function}

From statistical mechanics, we learn that the equilibrium distribution
function in a primed, co-moving frame moving with the particle mean-velocity will be (up to a normalization) $e^{-\beta H'}$, where $\beta>0$ is a scalar and $H'$ is the Hamiltonian in the co-moving frame. Consider the non-relativistic, free-streaming Hamiltonian, $H' = \frac{1}{2}g^{ij}p'_i p'_j$.\footnote{For the relativistic treatment of the equilibrium, see \cite{johnson2025moment}.} By transforming from the co-moving frame to the unprimed lab frame, the momentum transforms as $\mvec{p}' = \mvec{p} - \mvec{u}$, and hence the equilibrium distribution
function in the lab frame is
\begin{align}
  f_M(\mvec{x},\mvec{p})
  &=
  \frac{n}{(2\pi v_{th}^2)^{d/2}}
    \exp\left[
    -(\mvec{p}-\mvec{u})\cdot(\mvec{p}-\mvec{u})/(2
    v_{th}^2) 
    \right]  \notag \\
  &=
  \frac{n}{(2\pi v_{th}^2)^{d/2}}
  \exp\left[-g^{ij}(p_i-u_i) (p_j-u_j)/(2 v_{th}^2) \right]. \label{eq:lab_frame_eq}
\end{align}
(From this point onward, we will use canonical coordinates, so the
expressions below are only applicable to that case). Here, $n$ is the
number density, $\mvec{u} = u_i\dbasis{i}$ is the mean velocity, $\mvec{p} = p_i\dbasis{i}$ is the covariant momentum coordinate, and
$v_{th}$ is the thermal speed of the particles. The quantities that appear in the equilibrium can be computed
by taking moments of the distribution function. Integrating over the
momentum space gives (see equation \ref{eq:I0-multivariate-gauss} in the
Appendix)
\begin{align}
  \int_{-\infty}^\infty f_M \thinspace d\mvec{p}
  =
  n \sqrt{g}\label{eqn:density_int_dp}
\end{align}
where the momentum element is defined as, $d\mvec{p} \equiv dp_1 dp_2\ldots dp_d$. The Jacobian factor, $\sqrt{g}$, arises because we are integrating over covariant momentum variables in $d\mvec{p}$. The invariant volume element in momentum space is 
\begin{align}
    d\mathcal{V}_p = \frac{1}{\sqrt{g}}d\mvec{p}.
\end{align}
Hence,
\begin{align}
  \int_{-\infty}^\infty f_M \thinspace d\mathcal{V}_p 
  =
  n.\label{eqn:invariant_density_integral}
\end{align}
However, equation \ref{eqn:invariant_density_integral} is not as useful as equation \ref{eqn:density_int_dp} since the distribution is a function of covariant momentum $p_i$. In terms of $d\mvec{p}$, momentum takes the form
\begin{align}
  \int_{-\infty}^\infty \mvec{p} f_M \thinspace d\mvec{p}
  = n\sqrt{g}\mvec{u}.
\end{align}
Further, we have the
pressure-tensor (see equation \ref{eq:I2-multivariate-gauss})
\begin{align}
  \mvec{P}_M
  =
  \int_{-\infty}^\infty \mvec{c}\otimes\mvec{c} f_M \thinspace d\mvec{p}
  =
  n\sqrt{g} v_{th}^2  \mvec{g}
\end{align}
where $\mvec{c} = \mvec{p}-\mvec{u}$, and $\mvec{g}$ is the metric
tensor, $\mvec{g}(\mvec{a},\mvec{b}) = \mvec{a}\cdot\mvec{b}$.  Taking
the trace of this, we get
\begin{align}
  \Tr \mvec{P}_M 
  =
  \int_{-\infty}^\infty \mvec{c}\cdot\mvec{c} f_M \thinspace d\mvec{p}
  =
  d n v_{th}^2 \sqrt{g}
\end{align}
where recall that $\Tr \mvec{g} = g^{ij}g_{ij} = d$ is the dimension of the space.

We note that these normalizations of the distribution function imply
that the moments of the kinetic equation will be evolution
equations for the quantities $n \sqrt{g}$, etc, and not $n$
itself. This is a natural outcome of the choice of canonical
coordinates in which the volume factor appears automatically without
having to explicitly multiply or divide by it.

Collisions drive the particles to an equilibrium state. The
approach to equilibrium can be complex and is usually described
by Boltzmann or Fokker-Planck type collision operators. In this paper
we will assume a simpler model, the BGK collision operator \cite{bgk1954}, as an
approximation to more detailed collisional physics. In this operator
the collisions simply relax the distribution function to an
equilibrium:
\begin{align}\label{eq:BGK_op}
  C[f] = -\nu [f - f_M(n,\mvec{u},v_{th}) ]
\end{align}
where $\nu$ is a collision frequency (potentially depending on moments
of $f$ but not on momentum coordinates), and $f_M$ is the equilibrium
distribution function. The input moments needed to construct the
equilibrium are computed from the distribution function as
\begin{align}
  n\sqrt{g} &= \int_{-\infty}^\infty f \thinspace d\mvec{p} \label{eq:Jn} \\
  n\sqrt{g}\mvec{u} &= \int_{-\infty}^\infty \mvec{p} f \thinspace d\mvec{p} \label{eq:Jnu_i} \\
  d n v_{th}^2 \sqrt{g} &= \int_{-\infty}^\infty \mvec{c}\cdot\mvec{c} f \thinspace d\mvec{p}. \label{eq:v_th^2}
\end{align}
These moments ensure that the collision operator conserves particles,
momentum, and energy.

\section{Modifications for Motion on Axisymmetric Rotating Surfaces} \label{sec:sphere_theory}

\subsection{Hamiltonian of a Uniformly Rotating Spherical Surface}
To study the utility of the schemes developed in this paper to
properly account for additional effects, we consider the modifications to
the approach outlined above to include constant rotation of the
manifold. Consider a sphere rotating with constant
angular velocity. Let $\Omega$ be the angular speed and the $z$-axis,
$\cbas{3}$, be the axis around which the sphere rotates. Then, using
standard spherical coordinates, a point on a unit sphere has the
position vector
\begin{align}
  \mvec{x} = 
  \sin\theta 
  \big[ 
  \cos(\phi+\Omega t) \cbas{1} + \sin(\phi + \Omega t) \cbas{2} 
  \big]
  + \cos\theta \cbas{3}.
\end{align}
Here $\theta \in [0,\pi]$ is the polar angle, $\phi \in [0,2\pi]$ is
the azimuthal angle, and $\cbas{i}$, $i=1,2,3$, are the basis vectors
in the global Cartesian coordinates of the embedding space. From this position vector, we can compute the
tangent vectors
\begin{align}
  \basis{\theta} &= \pfrac{\mvec{x}}{\theta}
  = 
  \cos\theta 
  \big[ 
  \cos(\phi+\Omega t) \cbas{1} + \sin(\phi + \Omega t) \cbas{2} 
  \big]
  - \sin\theta \cbas{3} \\
  \basis{\phi} &= \pfrac{\mvec{x}}{\phi}
  = 
  \sin\theta 
  \big[ 
  -\sin(\phi+\Omega t) \cbas{1} + \cos(\phi + \Omega t) \cbas{2} 
  \big].
\end{align}
These are then used to compute the covariant components of the metric
tensor $g_{\theta\theta} = \basis{\theta}\cdot\basis{\theta} = 1$,
$g_{\phi\phi} = \basis{\phi}\cdot\basis{\phi} = \sin^2\theta$ and
$g_{\theta\phi} = g_{\phi\theta} = \basis{\phi}\cdot\basis{\theta} =
0$. The effect of rotation is hence not
apparent in the spatial metric. Therefore, if we were to use the Hamiltonian
$H = \mvec{p}\cdot\mvec{p}/2$, the rotation would not modify the
geodesics. 

Hence, to account for the rotation, we need to modify the
Hamiltonian. We start with the Lagrangian
\begin{align}
  L 
  = \frac{1}{2} \dot{\mvec{x}}\cdot\dot{\mvec{x}}
  = \frac{1}{2}
  \big[ 
  \dot{\theta}^2 + \sin^2\theta(\dot{\phi}+\Omega)^2
  \big].
\end{align}
From this, we compute the canonical momentum
\begin{align}
  p_\theta &= \pfrac{L}{\dot{\theta}} = \dot{\theta} \\
  p_\phi &= \pfrac{L}{\dot{\phi}} = \sin^2\theta(\dot{\phi}+\Omega).
\end{align}
The Hamiltonian is then constructed using the Legendre transform
\begin{align}
  H = p_\theta\dot{\theta} + p_\phi\dot{\phi} - L.
\end{align}
Eliminating the time-derivatives of the coordinates in favor of the
canonical velocities, we finally get
\begin{align}\label{eqn:rot_sphere_hamil}
  H = \frac{1}{2}
  \bigg(
  p_\theta^2 + \frac{p_\phi^2}{\sin^2\theta}
  \bigg)
  -
  \Omega p_\phi.
\end{align}
This is the modified Hamiltonian that accounts for the rotation of the
sphere. As we have used canonical coordinates, the Poisson bracket
remains the canonical one.

An effect that can be treated similarly is frame-dragging, which 
is a general relativistic effect that occurs around a
spinning object, in which a particle gets ``dragged'' along the
direction of the spin. This is a purely relativistic effect: in
Newtonian gravity the gravitational force outside a spherically
symmetric object does not depend on the angular momentum of the
object. Yet, in relativity, the motion of a test particle does
depend on the angular momentum. This effect in stationary spacetime manifests in the Hamiltonian as a term analogous to the constant spin case here, 
where a spatially dependent ``shift vector'' replaces the constant scalar spin $\Omega$.

\subsection{Equilibrium Modifications, Geodesics and Additional Forces}
A generalization of the Hamiltonian to include background flows and static potentials represents an extended range of systems, including the uniformly rotating sphere. With static background flows $v_i(x^i)$ and potential $V(x^i)$, the Hamiltonian takes the form
\begin{align}
    H(x^i,p_i) = \frac{1}{2} g^{ij} (p_i - v_i) (p_j - v_j) + V. \label{eqn:extended_hamil_def_with_static_fields}
\end{align}
which includes, by completing the square, the uniformly rotating surface of a sphere Hamiltonian equation \ref{eqn:rot_sphere_hamil}. The equilibrium for this Hamiltonian must therefore include these background flows, which becomes
\begin{align}
    f_M(x^i,p_i) = \frac{n}{(2\pi v_{th}^2)^{d/2}} \exp ( -g^{ij} (p_i - v_i - u_i) (p_j - v_j - u_j) / (2v_{th}^2)). \label{eqn:equilibrium_with_static_potential}
\end{align}
where, in this case, the static potential has been absorbed into the definition of the density. As we will see in the subsequent section, to be discretely consistent with the characteristic velocities, we must use the Hamiltonian and its gradients to compute moments, as it is these gradients of $H$ that determine the characteristics which are used. In the continuous limit, these integral kernels for moments are identical, but they are not necessarily so when $H$ is represented discretely. To facilitate the transition into the next section, which describes the discrete scheme, the moments are written here in terms of the gradients of the Hamiltonian
\begin{align}
    n \sqrt{g} &= \int_{-\infty}^{\infty} f d\mvec{p} \label{eqn:mom_sphere_1}\\
    n \sqrt{g} u^i &= \int_{-\infty}^{\infty} \pfrac{H}{p_i} f d\mvec{p} \label{eqn:mom_sphere_2}\\
    d n \sqrt{g} v_{th}^2 &= 2 \int_{-\infty}^{\infty} H f d\mvec{p} - n \sqrt{g} (g^{ij}u_i u_j + 2 V)\label{eqn:mom_sphere_3}
\end{align}
Where the momentum components are related via the metric, $u_i = g_{ij} u^j$. 

In the case of a uniformly rotating unit sphere, the background flows and static potential take the form
\begin{align}
    v_\theta &= 0\\
    v_\phi &= \Omega \sin^2\theta \\
    V &= - \frac{1}{2} \Omega^2 \sin^2 \theta.
\end{align}
Using these expressions in the equilibrium equation that includes background flows and the static potential, equation \ref{eqn:equilibrium_with_static_potential}, as well as rewriting the Hamiltonian in the form  \ref{eqn:extended_hamil_def_with_static_fields},
\begin{align}
    H = \frac{1}{2} \left( p_\theta^2 + \frac{ (p_\phi - \Omega \sin^2 \theta)^2}{\sin^2 \theta} \right) - \frac{1}{2} \Omega^2 \sin^2 \theta.
\end{align}
gives an alternative form the Hamiltonian for the rotating sphere. A straightforward calculation of the equations of motion yields the motion along great circles (seen in the \ref{app:geodesics}), with additional terms for Coriolis and centrifugal forces
\begin{align}
    \ddot{\theta} 
    &=
    \underbrace{
     \dot{\phi}^2 \sin \theta \cos \theta 
     }_{\text{Great Circle}}
     +
     \underbrace{
       \Omega \dot{\phi}  \sin 2 \theta 
     }_{\text{Coriolis}}
     +
     \underbrace{
       \Omega^2 \sin \theta \cos \theta 
     }_{\text{Centrifugal}} 
     \\
    \ddot{\phi} 
    &=
    -
    \underbrace{
      2 \dot{\theta} \dot{\phi} \cot \theta 
     }_{\text{Great Circle}}
     -
     \underbrace{
        2 \Omega \dot{\theta} \cot \theta 
     }_{\text{Coriolis}}.
\end{align}
In these equations, the geodesic motion on the surface of a sphere, as well as additional forces, are directly incorporated into the advection equation for the system. For a sphere rotating at a uniform constant velocity including collisions, the kinetic equation takes the form
\begin{align}
    \pfrac{f}{t} 
    + 
    p_\theta \pfrac{f}{\theta} 
    + 
    \left( \frac{p_\phi}{\sin ^2 \theta} - \Omega \right)\pfrac{f}{\phi} 
    +
    \frac{p_\phi^2 \cos \theta}{\sin^3 \theta}\pfrac{f}{p_\theta} 
    =
    - \nu (f - f_M) \label{eqn:advective_rot_sphere}
\end{align}
In the limit where collisions dominate, $\lim \nu  \rightarrow \infty$, the result is that the distribution acts as a fluid, and the BGK operator replaces $f$ with $f_M$ faster than the advective timescales. Integrals of the distribution, equations \ref{eqn:mom_sphere_1}, \ref{eqn:mom_sphere_2}, \ref{eqn:mom_sphere_3}, taken over the advection equation \ref{eqn:advective_rot_sphere}, yield evolution equations for the moments of the equilibrium: Euler's equations. The integrals yield the same result as Euler's equations on the surface of a sphere in a rotating reference frame,
\begin{align}
    \pfrac{n}{t} 
    +
    \frac{1}{\sin \theta} \pfrac{}{\theta} \left( n \sin\theta u_\theta\right) 
    +
    \pfrac{}{\phi} \left( n \left(\frac{u_\phi}{\sin^2 \theta} - \Omega \right) \right) 
    =
    0  \label{eqn:continuity_sphere}
\end{align}
\begin{align}
    \pfrac{u_\theta}{t} 
    +
    u_\theta \pfrac{u_\theta}{\theta}
    +
    \frac{1}{n} \pfrac{}{\theta} \left( n v_{th}^2 \right) 
    +
    \left( \frac{u_\phi}{\sin^2 \theta}  - \Omega \right) \pfrac{u_\theta}{\phi}
    -
    u_\phi^2 \frac{\cos \theta}{\sin ^3 \theta} 
    =
    0 \label{eqn:p_theta_sphere}
\end{align}
\begin{align}
    \pfrac{u_\phi}{t} 
    +
    u_\theta \pfrac{u_\phi}{\theta}
    +
    \frac{1}{n} \pfrac{}{\phi} \left( n v_{th}^2 \right) 
    +
    \left( \frac{u_\phi}{\sin^2 \theta}  - \Omega \right) \pfrac{u_\phi}{\phi}
    =
    0 \label{eqn:p_phi_sphere}
\end{align}
Alternatively, expressed in terms of the relative velocity, $\tilde{u}_\phi = u_\phi - \Omega \sin^2 \theta$, rather than the total velocity $u_\phi$, the terms can be easily identified as either advective, pressure, centrifugal, or Coriolis:
\begin{align}
    \pfrac{n}{t} 
    +
    \frac{1}{\sin \theta} \pfrac{}{\theta} \left( n \sin\theta u_\theta \right) 
    +
    \pfrac{}{\phi} \left(\frac{n\tilde{u}_\phi}{\sin^2 \theta}\right)
    =
    0  \label{eqn:continuity_sphere_rel_vel}
\end{align}
\begin{align}
    \pfrac{u_\theta}{t} 
    +
    \underbrace{
        u_\theta \pfrac{u_\theta}{\theta}
        +
        \frac{\tilde{u}_\phi}{\sin^2 \theta} \pfrac{u_\theta}{\phi}
    }_{\text{Advective}}
    +
    \underbrace{
        \frac{1}{n} \pfrac{}{\theta} \left( n v_{th}^2 \right) 
    }_{\text{Pressure}}
    -
    \frac{\cos \theta}{\sin ^3 \theta} 
    \bigg (
        \underbrace{
            \tilde{u}_\phi^2 
            +
            \Omega^2 \sin ^4 \theta
        }_{\text{Centrifugal}}
        +
        \underbrace{
            2\Omega\tilde{u}_\phi \sin^2 \theta
        }_{\text{Coriolis}}
    \bigg )
    =
    0 \label{eqn:p_theta_sphere_rel_vel}
\end{align}
\begin{align}
    \pfrac{\tilde{u}_\phi}{t} 
    +
    \underbrace{
        u_\theta \pfrac{\tilde{u}_\phi}{\theta}
        +
        \frac{\tilde{u}_\phi}{\sin^2 \theta} \pfrac{\tilde{u}_\phi}{\phi}
    }_{\text{Advective}}
    +
    \underbrace{
        \frac{1}{n} \pfrac{}{\phi} \left( n v_{th}^2 \right) 
    }_{\text{Pressure}}
    +
    \underbrace{
        2\Omega u_\theta \sin \theta \cos \theta
    }_{\text{Coriolis}}
    =
    0. \label{eqn:p_phi_sphere_rel_vel}
\end{align}
This type of procedure establishes the formal connection between kinetic theory and fluid theories for neutrals in the highly collisional limit. Additional effects, such as contributions from electromagnetic fields, can similarly be considered. In addition to single particle motion, the highly collisional regime gives us another useful limit against which we can test the canonical Poisson bracket solver. 
% -*- latex -*-
\section{A Discontinuous Galerkin Scheme and Iterative Corrections for
  Collisions}\label{sec:scheme}

\subsection{The Discontinuous Galerkin Scheme}
Here, we consider the properties of the general discrete Poisson Bracket updater with an arbitrary, potentially non-canonical Poisson Tensor, $\gvec{\Pi}$, as well as the special case where the Poisson Tensor equals the antisymmetric unit tensor, $\gvec{\Pi} = \gvec{\sigma}$, i.e., the canonical case\footnote{For discussions of the numerical properties of related Hamiltonian systems such as Vlasov-Poisson, incompressible Euler, and gyrokinetic schemes, see \cite{hakim2019discontinuous}.}. We will demonstrate, in general, that the updater conserves the total number of particles as well as the total energy. Further, we show that the canonical case conserves the $L_2$ norm for central fluxes or monotonically decays the $L_2$ norm when using upwind fluxes. For the analysis, we only consider regions of the geometry away from coordinate singularities, where the coordinates and the associated basis are well defined.

We begin by discretizing the phase space into cells, $V_n$, where the index $n$ spans $n = 1 \ldots N$, and $N$ represents the total number of phase space cells. Collectively, these cells form a phase-space mesh $\mathcal{T}$. To discretize the evolution, equation \ref{eq:f-hamil}, we introduce a piecewise-polynomial space, $\mathcal{V}_h^p$, to which the distribution function, $f(t,z^i)$, belongs. The subscript $h$ represents quantities that belong discretely to the modal space. The piecewise polynomial space exists on the phase-space mesh $\mathcal{T}$,
\begin{align}
    \mathcal{V}_h^p = \{ w:w|_{V_n} \in \mvec{P}^p,\forall V_n \in \mathcal{T} \}.
\end{align}
This expression is read as a piecewise polynomial space, $\mathcal{V}_h^p$, which contain $w(\mvec{z})$ as a  polynomial functions in each cell $V_n$ and belongs to the space of linear combinations of basis polynomials, $\mvec{P}^p$ for the entire phase space mesh. Lowercase $\mvec{z}$ is the phase space within each cell. The distribution, represented as piecewise polynomials, is therefore allowed to be discontinuous. The polynomial space $\mvec{P}^p$ in each cell is spanned by the Lagrange tensor basis or the Serendipity basis \cite{arnold2011serendipity}, or some other suitable set of functions. All simulations presented in this paper use the Serendipity basis.

In contrast to the discontinuous representation of the distribution function, the Hamiltonian, as we show below for discrete energy conservation,  must be continuous. The discrete continuous space is denoted as $\mathcal{W}_{0,h}^p$,
\begin{align}
    \mathcal{W}_{0,h}^p =  \mathcal{V}_h^p \cap C^0(\mvec{Z})
\end{align}
where $\mvec{Z}$ is the entire phase-space domain and $C^0(\mvec{Z})$ is the set of all continuous functions. Further, for non-canonical systems, the discrete Poisson tensor components, $\Pi^{ij}_h$, must also be continuous, $\Pi^{ij}_h \in \mathcal{W}_{0,h}^p$ to ensure discrete energy conservation. We note that the antisymmetry property of the Poisson tensor is preserved through the projection of the components onto the modal basis, $\Pi^{ij}_h = - \Pi^{ji}_h$.

The DG scheme can then be formulated as finding the time evolution of $\mathcal{J} f$ when projected onto the basis, that is $(\mathcal{J}f)_h \in \mathcal{V}_h^p$. To find the evolution of $(\mathcal{J}f)_h$ we begin with the Vlasov equation \ref{eq:f-hamil} and work to rewrite it in the discrete-weak form. We multiply the continuous Vlasov equation by $\mathcal{J}$ and move $\mathcal{J}$ into the time derivative since we are strictly examining cases that the geometry is independent of time. Using the divergence form of the bracket from equation \ref{eq:pb-div}, the Vlasov equation becomes an evolution equation for $\mathcal{J}f$,
\begin{align}
    \pfrac{}{t} \left( \mathcal{J} f \right)
    + 
    \pfraca{z^i}
    \left(
    \mathcal{J} f \Pi^{ij} \pfrac{H}{z^j}
    \right) = 0.\label{eqn:Jf_evo}
\end{align}
To derive the discrete weak form, we multiply by the test function $w$ and integrate over phase space over each cell $V_n$ to get 
\begin{align}\label{eqn:discrete_weak_form_pre_IBP}
    \int_{V_n} w 
    \left( 
    \pfrac{}{t} \left( \mathcal{J} f \right)
    + 
    \pfraca{z^i}
    \left(
    \mathcal{J} f \Pi^{ij} \pfrac{H}{z^j} 
    \right)  
    \right) d\mvec{z}
    = 0.
\end{align}

Integration by parts on equation \ref{eqn:discrete_weak_form_pre_IBP} moves the derivatives off of $\mathcal{J}f\Pi^{ij} \pfrac{H}{z^j}$ and onto the test functions, $w$. Applying divergence theorem gives a surface flux term.\footnote{The $d\mvec{z}$ computational coordinates are rectangular. Therefore, when we apply the divergence theorem, we do not obtain extra Jacobian factors inside the integrals.} The resulting discretized equation is referred to as the discrete-weak form of the Vlasov equation, 
\begin{align}
      \int_{V_n} w  
      \frac{\partial}{\partial t} \left( \mathcal{J} f \right)_h d\mvec{z}
      &+
      \sum_m^{2d} \int_{S_n^m} 
      w^- \mathcal{J}_h \{ z^m, H_h\}_h \hat{f} d\mvec{z}^\slashed{m}\bigg |^{z_n^m+\frac{dz_n^m}{2}}_{z_n^m-\frac{dz_n^m}{2}} \\ \notag
      &-
      \int_{V_n} 
      \{ w,H_h \}_h(\mathcal{J} f)_h d\mvec{z} 
      = 
      0.
      \label{eqn:discrete_weak_form}
\end{align}
Where we have projected the continuous functions onto the modal basis, which is indicated by the subscript $h$ and where $z^m_n$ is the $n$-cell centered location, and $dz^m_n$ is the difference in direction of $m$. The sum over $m$ is the sum over all of the phase space cell faces. $d\mvec{z} = dz^1...dz^{2d}$ is the computational-cell volume element, and $d\mvec{z}^\slashed{m} = dz^1...dz^{2d}/dz^m$ is the cell surface area element in the direction of $m$. The characteristic velocity component in direction $m$ is $\alpha_h^m = \{ z^m, H_h \}_h$. The superscripts $w^{+}$ and $w^{-}$ denote functions evaluated just outside and inside the cell's surface, respectively. The grouping of $(\mathcal{J} f)_h$ in parentheses in equation \ref{eqn:discrete_weak_form} indicates that this term is combined and projected onto $\mathcal{V}_h^p$. 

The volume term has been discretized in terms of the discretely projected Poisson bracket, and by further utilizing weak identities, it has been combined into a single projection
\begin{align}
    \{ w, H_h\}_h = \left( \pfrac{w}{z^i} \Pi^{ij}_h\pfrac{H_h}{z^j} \right)_h.
\end{align}

In addition, we have introduced a flux function in the surface term, $\hat{f}$. The use of the flux function ensures flux agreement across the cell boundary, which is required for maintaining discrete conservation laws. For the numerical flux function, we use a Lax flux 
\begin{align}\label{eqn:general_flux_function}
    \alpha_h^m \hat{f}_h
    =
    \alpha_h^m \left( \frac{ f^+_h + f^-_h }{2} \right)
    -
    \frac{\lambda}{2}
    \left(
    f^+_h - f^-_h 
    \right)
\end{align}
where we consider the two cases of a central flux $\lambda = 0$, or an upwind flux $\lambda = |\alpha_h^m|$. We have isolated $f$ and $\mathcal{J}$ for the surface term by the weak equality
\begin{align}
    (f \mathcal{J})_h \doteq f_h \mathcal{J}_{h}.
\end{align}

The volume integral terms are treated by precomputing the volume integral and unrolling the loops with a computer algebra solver (CAS) to avoid matrix multiplication. Details on eliminating aliasing errors and removing matrix multiplication and quadrature methods in the algorithm we present here are described in\cite{hakim2020alias}. The time advancement is computed discretely using a strong stability-preserving (SSP) third-order Runge-Kutta scheme\cite{shu1988efficient}.

% Adapted from AH/GH/ES/NM DG schemes Hamil. Evo.
\subsection{Discrete Conservation and Stability Properties}

% 1. Comment about how f has to go to zero at all boundaries, periodic may not be applicable.
Hamiltonian and Poisson tensor continuity guarantees that the characteristic velocity normal component to a phase-space cell is continuous. We can thus equate components on both sides of the cell interface, $\alpha^{m+}_h = \alpha^{m-}_h = \alpha^{m}_h$. We prove this continuity of the normal component of the characteristic velocity now. 

\begin{lemma}
Components of the characteristic velocity normal to the face of a phase-space cell are continuous. \label{lemma:normal_cont}
\end{lemma}
\begin{proof}
For a discrete Hamiltonian system, the Poisson bracket is given by,
\begin{align}
    \{f_h, g_h \} = \pfrac{f_h}{z^i} \Pi^{ij}_h \pfrac{g_h}{z^j}
\end{align}
where $\Pi^{ij}_h$ is the discretely represented Poisson tensor. The characteristic velocity, $\alpha^i_h = \{z^i,H_h\}_h$, can be written as $\alpha^i_h = \Pi^{ij}_h \pfrac{H_h}{z^j}$. Let $n_i$ be the components of the unit vector normal to a cell surface. Contracting the normal vector with the characteristic speed,
\begin{align}
    n_i\alpha^i_h
    =
    n_i \Pi^{ij}_h \pfrac{H_h}{z^j}
    =
    \tau^j \pfrac{H_h}{z^j}, \label{eqn:n_dot_alpha}
\end{align}
where $\tau^j = n_i\Pi^{ij}_h$. Antisymmetry of the Poisson tensor gives
\begin{align}
    \tau^i n_i = \gvec{\tau} \cdot \gvec{n}= n_i\Pi^{ij}_hn_j = 0
\end{align}
This relation shows that $\gvec{\tau}$ is orthogonal to $\mvec{n}$ and tangent to the cell surface. Thus, $n_i\alpha^i_h$ depends only on gradients of the Hamiltonian that are tangent to the surface and not on gradients normal to the surface, which may be discontinuous.

Hence, since the discrete Hamiltonian belongs to $\mathcal{W}_{0,h}^p$, its tangential gradients are continuous. Likewise, the Poisson tensor belongs to $\mathcal{W}_{0,h}^p$ and is also continuous, hence,
\begin{align}
    n_i \alpha^i_h 
    = 
    n_i \Pi^{ij+}_h \pfrac{H_h^+}{z^j} 
    = 
    n_i\Pi^{ij-}_h \pfrac{H_h^-}{z^j}.
\end{align}
Therefore, the normal components of the characteristic velocity are also continuous across the cell interface.
\end{proof}

The importance of the continuity of the characteristics of the discrete system is that this property is needed to prove the conservation of total energy. In addition, this property is needed to prove the conservation or monotonic decay of the $L_2$ norm when the phase space is incompressible within the cell volume. Phase space incompressibility is defined as the characteristics being divergence free $\mvec{\nabla} \cdot \gvec{\alpha} = 0$, which for the general discrete system is written as
\begin{align}
    \mvec{\nabla} \cdot \gvec{\alpha}_h = \frac{1}{\mathcal{J}_h}\pfrac{}{z^i} \left( \mathcal{J}_h \Pi^{ij}_h \pfrac{H_h}{z^j} \right). \label{eqn:phase_space_incompressibility_general}
\end{align}

Expanding the divergence on the right hand side of equation \ref{eqn:phase_space_incompressibility_general} into two terms,
\begin{align}
    \frac{1}{\mathcal{J}_h}\pfrac{}{z^i} \left( \mathcal{J}_h \Pi^{ij}_h \pfrac{H_h}{z^j} \right)
    =
    \frac{1}{\mathcal{J}_h} \underbrace{
         \left( \mathcal{J}_h \Pi^{ij}_h \right)  \frac{ \partial^2  H_h}{ \partial z^i \partial z^j}
    }_{= 0}
    +
    \frac{1}{\mathcal{J}_h} \underbrace{ \pfrac{}{z^i} \left( \mathcal{J}_h \Pi^{ij}_h  \right) }_{\text{Liouville Identities}} \pfrac{H_h}{z^j}.
\end{align}
The first term on the right hand side is zero within the cell volume due to the antisymmetry of the Poisson tensor and the symmetry under swapping the indices of the partial derivatives of the Hamiltonian. The second term, however, is not necessarily zero due to the projections of the Jacobian factor and Poisson tensor onto the basis functions. Recall that prior to being projected onto bases, the underlined quantity obeys the Liouville identities; hence, this term is continuously zero. Discretely, the projections may not preserve the Liouville identities for the non-canonical system; therefore, we cannot generally prove phase space incompressibility in the non-canonical case.

However, phase space incompressibility is provable for the canonical case when $\gvec{\Pi} = \gvec{\sigma}$. This is formalized in the following lemma and proof.

\begin{lemma}
In the canonical case, the phase space is incompressible within the cell volume. \label{lemma:ps_incompress}
\end{lemma}
\begin{proof}
In general, the discrete divergence of the phase space of the Hamiltonian system is,
\begin{align}
    \frac{1}{\mathcal{J}_h}\pfrac{}{z^i} \left( \mathcal{J}_h \Pi^{ij}_h \pfrac{H_h}{z^j} \right).
\end{align}
This simplifies in the canonical case where $\mathcal{J}_h = 1$ and $\Pi^{ij}_h = \sigma^{ij}$ to,
\begin{align}
    \frac{1}{\mathcal{J}_h}\pfrac{}{z^i} \left( \mathcal{J}_h \Pi^{ij}_h \pfrac{H_h}{z^j} \right)
    =
    \pfrac{}{z^i} \left( \sigma^{ij} \pfrac{H_h}{z^j} \right). \label{eqn:phase_space_incompressibility}
\end{align}
where $\sigma^{ij}$ is the antisymmetric unit tensor. 

As $\sigma^{ij}$ is constant, it may be pulled out of the derivative. 
\begin{align}
    \pfrac{}{z^i} \left( \sigma^{ij} \pfrac{H_h}{z^j} \right)
    =
    \sigma^{ij}  \frac{ \partial^2  H_h}{ \partial z^i \partial z^j}
\end{align}
Within the volume, $H_h$ is composed of a polynomial basis and is therefore smooth within each cell. Thus, the second partials of the discrete Hamiltonian are symmetric under the swapping of the indices. Using the symmetric property of the partials and the antisymmetry, $\sigma^{ij} = - \sigma^{ji}$, and then relabeling, the right hand side of equation \ref{eqn:phase_space_incompressibility} becomes zero
\begin{align}
    \sigma^{ij}  \frac{ \partial^2  H_h}{ \partial z^i \partial z^j}
    =
    - \sigma^{ij}  \frac{ \partial^2  H_h}{ \partial z^i \partial z^j}
    = 0.
\end{align}
Hence, $\mvec{\nabla} \cdot \gvec{\alpha} = 0$, and the phase space is incompressible in the canonical limit.
\end{proof}

Looking forward, Lemma \ref{lemma:normal_cont} is required to prove the conservation law for the total energy. In addition to Lemma \ref{lemma:normal_cont}, Lemma \ref{lemma:ps_incompress} is then used to prove stability properties for the canonical system regarding the $L_2$ norm conservation for central fluxes, or the monotonic decay of the $L_2$ norm for up-winded fluxes.

For the subsequent proofs, we will assume the boundary surface terms vanish. Example boundary conditions which satisfy this assumption include periodic and reflecting wall boundary conditions in configuration space and zero flux boundary conditions in velocity space. 

\begin{proposition}
    The total number of particles is conserved exactly.
\end{proposition}
\begin{proof}
    Conservation of the total number of particles follows immediately upon replacing $w = 1$ in the discrete weak form, equation \ref{eqn:discrete_weak_form}, and summing over all cells, leading to the equation
    \begin{align}
        \pfrac{}{t}\sum_{V_n 
        \in
        \mathcal{T}} \int_{V_n} (\mathcal{J}f)_h d\mvec{z} = 0.
        \hfill 
    \end{align} 
    Where, by summing over all cells, the contribution from the surface integral term of the discrete weak form drops, as the integrand differs only in sign for the two cells sharing a face. The cancellation relies on the continuity of the flux function. In our case, this is ensured from Lemma \ref{lemma:normal_cont}. The outward normal differs between cells only by a sign, but its magnitude is shared. Finally, the volume term of the discrete weak form is zero, as the derivative of $w = 1$ is zero.
    \qedhere
\end{proof}

\begin{proposition}
    The spatial scheme conserves total energy exactly.
\end{proposition}
\begin{proof}
    Selecting the test function $w = H_h$ in the discrete weak form \ref{eqn:discrete_weak_form}, we have,
    \begin{align}
        \int_{V_n} H_h  
        \frac{\partial}{\partial t} \left( \mathcal{J} f \right)_h d\mvec{z}
        &+
        \sum_m^{2d} \int_{S_n^m} 
        H_h^- \mathcal{J}_h \{ z^m, H_h\}_h \hat{f} d\mvec{z}^\slashed{m}\bigg |^{z_n^m+\frac{dz_n^m}{2}}_{z_n^m-\frac{dz_n^m}{2}} \\ \notag
        &-
        \int_{V_n} 
        \{ H_h,H_h \}_h(\mathcal{J} f)_h d\mvec{z} 
        = 
        0.
        \label{eqn:discrete_weak_form_with_Hh}
    \end{align}
    The volume term vanishes because the Poisson tensor is antisymmetric, 
    \begin{align}
        \{H_h,H_h\}_h 
        = 
        \left( \pfrac{H_h}{z^i} \Pi^{ij} \pfrac{H_h}{z^j} \right)_h 
        = 
        0. 
    \end{align}
    Physically, the last term vanishes because the flow $\gvec{\alpha}_h$ is along contours of constant energy in phase space. Summing over all cells, the contribution from the surface integral term drops, as the integrand differs only in sign for the two cells sharing a face. Here, the cancellation relies on the continuity of $H_h$, normal characteristics, the flux function, and $\mathcal{J}_h$. The normal differs between cells only by a sign, but its magnitude is shared. Thus, we obtain a conservation of total energy,
    \begin{align}
        \sum_{V_n
        \in
        \mathcal{T}} \int_{V_n} H_h  \pfrac{}{t} 
        \left( \mathcal{J} f \right)_h  d\mvec{z} 
        = 0.
        \hfill 
    \end{align}\qedhere
\end{proof}

\begin{proposition}
    In the canonical case, the spatial scheme exactly conserves the $L_2$ norm of the distribution function when using a central flux, while the distribution-function $L_2$ norm monotonically decays when using an upwind flux.
\end{proposition}
\begin{proof}
    Substituting the test function $w = f_h$ into equation \ref{eqn:discrete_weak_form} in the canonical case where $\mathcal{J}_h = 1$, $(\mathcal{J}f)_h = f_h$, and $\gvec{\Pi} = \gvec{\sigma}$, yields
    \begin{align}
          \int_{V_n} f_h  
          \frac{\partial f_h}{\partial t}  d\mvec{z}
          &+
          \sum_m^{2d}  \int_{S_n^m} 
          f_h^- \{ z^m, H_h\}_h \hat{f} d\mvec{z}^\slashed{m}\bigg |^{z_n^m+\frac{dz_n^m}{2}}_{z_n^m-\frac{dz_n^m}{2}}
          -
          \int_{V_n} 
          \{ f_h,H_h \}_h f_h d\mvec{z} 
          = 
          0.
          \label{eqn:discrete_weak_form_w_eq_fh}
    \end{align}
    In the volume term, the bracket can be replaced using weak identities under the integral to break apart the projection of the bracket into a gradient of $f_h$ and $\alpha_h^i$, 
    \begin{align}
          \int_{V_n} 
          \{ f_h, H_h \}_h f_h d\mvec{z}
          = 
          \int_{V_n} 
          \left( \pfrac{f_h}{z^i} \right) \alpha_h^i f_h d\mvec{z}.
    \end{align}
   This regrouping is allowed in the canonical case by weak identities because the Poisson tensor, $\gvec{\Pi} = \gvec{\sigma}$, is a constant rather than an expansion in the basis. Then, the last term of equation \ref{eqn:discrete_weak_form_w_eq_fh} can be rewritten as
   \begin{align}
       \int_{V_n} 
       \left( \pfrac{f_h}{z^i} \right) \alpha_h^i f_h d\mvec{z}
       &=
       \int_{V_n} 
       \pfrac{}{z^i} \left(\frac{f_h^2}{2} \alpha_h^i \right)  d\mvec{z}
       -
       \int_{V_n} 
       \frac{f_h^2}{2}\pfrac{\alpha_h^i}{z^i}  d\mvec{z}. \label{eqn:div_alpha_step_of_L2_norm_proof}
   \end{align}
   
   Applying phase space incompressibility from Lemma \ref{lemma:ps_incompress} for the canonical case, the volume term is zero
   \begin{align}
       \pfrac{\alpha^i_h}{z^i} 
       = 
       \sigma^{ij} \frac{\partial^2 H_h}{\partial z^i \partial z^j}
       =
       0.
   \end{align}
   This means the last term of equation \ref{eqn:div_alpha_step_of_L2_norm_proof} is zero. Applying the divergence theorem to the first term of equation \ref{eqn:div_alpha_step_of_L2_norm_proof}, and combining it with the surface term of equation \ref{eqn:discrete_weak_form_w_eq_fh}, we have
    \begin{align}
        \frac{\partial}{\partial t} 
        \int_{V_n}
        \frac{1}{2} f_h^2 d\mvec{z}
        +
        \sum_m^{2d} \int_{S_n^m} 
        \alpha_h^m f_h^- 
        \left(
        \hat{f} - \frac{1}{2}f_h^- 
        \right) d\mvec{z}^\slashed{m}\bigg |^{z_n^m+\frac{dz_n^m}{2}}_{z_n^m-\frac{dz_n^m}{2}}
        = 
        0.
    \end{align}
    
    With the \textit{central flux}, $\hat{f} = (f_h^+ + f_h^-)/2$, we are left with, 
    \begin{align}
        \frac{\partial}{\partial t} \int_{V_n} \frac{1}{2} f_h^2 d\mvec{z}
        +
        \sum_m^{2d} \int_{S_n^m} \frac{1}{2}\alpha_h^m f_h^-f_h^+ d\mvec{z}^\slashed{m} \bigg |^{z_n^m+\frac{dz_n^m}{2}}_{z_n^m-\frac{dz_n^m}{2}}
        =
        0.
    \end{align}
    Summing over all phase-space cells, the $L_2$ norm is conserved exactly,
    \begin{align}
           \pfrac{}{t} \sum_{V_n 
           \in
           \mathcal{T}} \int_{V_n} \frac{1}{2} f_h^2 d\mvec{z}
           = 
           0.
    \end{align}
    
    With the \textit{upwind flux}, $\lambda = |\alpha^m_h|$, used in equation \ref{eqn:general_flux_function}, the second term of the integrand becomes,
    \begin{align}
        \alpha^m_h f_h^- 
        \left(
        \hat{f} - \frac{1}{2}f_h^- 
        \right)
        =
        \frac{1}{2} \alpha^m_h f_h^- f_h^+ 
        - 
        \frac{1}{4}| \alpha^m_h| 
        \left( 
        (f_h^+)^2 - (f_h^-)^2 
        \right) 
        +
        \frac{1}{4}|\alpha^m_h| ( f_h^+ - f_h^-)^2.
    \end{align}
    Summing over all phase-space cells eliminates the first and second terms,
    \begin{align}
        \pfrac{}{t} \sum_{V_n \in \mathcal{T}} 
        \int_{V_n} 
        \frac{1}{2} f_h^2 d\mvec{z} 
        =
        - \frac{1}{4}\sum_{V_n \in \mathcal{T}} 
        \sum_m^{2d} \int_{S_n^m} 
        |\alpha^m_h| ( f_h^+ - f_h^-)^2 d\mvec{z}^\slashed{m}\bigg |^{z_n^m+\frac{dz_n^m}{2}}_{z_n^m-\frac{dz_n^m}{2}}
    \end{align}
    Since the right hand side is less-than or equal to zero, the global $L_2$ norm decays monotonically.
\end{proof}

%% GJ (11/26/2025) SEE NOTES: The step where we differentiate both sides of the inequality and sum over phase space is not valid. 

% \begin{proposition}
%     In the canonical case, if the discrete distribution function $f_h$ remains positive definite, then the discrete scheme grows the discrete entropy monotonically,
%     \begin{align}
%         \frac{d}{dt}\sum_{V_n \in \mathcal{T}}\int_{V_n}-f_h \ln{(f_h)}d\mvec{z} \geq 0.
%     \end{align}
% \end{proposition}
% \begin{proof}
%     We have the bound $\ln(f_h) \leq f_h - 1$ as long as $f_h > 0$. Multiplying by $-f_h$ gives us the inequality,
%     \begin{align}
%         - f_h \ln{(f_h)} \geq -f_h^2 + f_h.
%     \end{align}
%     The left-hand side is the discrete entropy of the Vlasov-equation. Integration over the phase-space cell, summation over all cells, and taking the time derivative of this inequality yields,
%     \begin{align}
%         \frac{d}{dt}\sum_{V_n \in \mathcal{T}}  
%         \int_{V_n} 
%         -f_h \ln{(f_h)}d\mvec{z}
%         \geq
%         \frac{d}{dt} \sum_{V_n \in \mathcal{T}}  
%         \int_{V_n}  
%         \left( -f_h^2 + f_h \right) d\mvec{z}.
%     \end{align}
%     Where the right-hand side $f_h$ term vanishes due to particle conservation, and the $-f_h^2$ term grows monotonically following the proof of the decay of the $L_2$ norm.
% \end{proof}

\subsection{The General $L_2$ Norm and Some Useful Weak Identities}

The discrete $L_2$ stability proof depends on the weak identity
\begin{align}
    \int_{V_n} w_h f d\mvec{z} = \int_{V_n} w_h f_h d\mvec{z}\label{eqn:weak_id_1}
\end{align} 
for functions $w$ and $f_h$ in $\mathcal{V}_h^p$. Here, $f$ (with no subscript) is a general function that may not lie within $\mathcal{V}_h^p$. But equation \ref{eqn:weak_id_1} demonstrates that only the projections of $f$ within the basis contribute to the integral. Further, equation \ref{eqn:weak_id_1} is what allowed us to separate the volume term's $\alpha_h^i$ from the gradient of $f_h$ to proceed with the proof via,
\begin{align}
     \int_{V_n} w_h (f g)_h d\mvec{z} = \int_{V_n} w_h f_h g_h d\mvec{z}, \label{eqn:weak_id_2}
\end{align} 
which is a specific case of equation \ref{eqn:weak_id_1}, where $(fg)_h$ and $g_h$ are also functions in $\mathcal{V}_h^p$. Unfortunately, the $L_2$ norm for the general case cannot be proven to be bounded. The reason is that while $(\mathcal{J}f)_h \doteq \mathcal{J}_hf_h$ are weakly equal under integration via equation \ref{eqn:weak_id_2} only when these are the only two factors. Otherwise, we cannot separate them into their own projections without changing the value of the volume term integral. 

The second weak identity prevents arbitrary grouping
\begin{align}
    \int_{V_n} w_h (f g)_h h_h d\mvec{z} \neq \int_{V_n} w_h f_h g_h h_h d\mvec{z} 
\end{align} 
for functions $w$, $f_h$, $g_h$, and $h_h$ in $\mathcal{V}_h^p$. It is this identity that prevents the breakup of the $(\mathcal{J}f)_h$ into $\mathcal{J}_hf_h$ under the volume term integral, where in the non-canonical case you have an extra function for $\gvec{\Pi}_h$ and $\mathcal{J}_h$. Secondly, we do not have a proof of incompressible phase space for the general case, which means we cannot write this term as a divergence, as was possible for the canonical case. While non-canonical brackets are interesting for providing alternative ways of writing coordinates, the system will need some reformulation to achieve provable $L_2$ norm stability for this case, which is outside the scope of this paper. We mention that in the presence of even weak collisions, the $L_2$ norm will decay from the collsions, hence increasing the stability of the solver even in the non-canonical case. For more details on weak identities, see \ref{app:weak_ids}

\subsection{Discrete Time Advancement and Collisions}
% Talk about SSP-RK3 effects
As we have proven, the overall discrete scheme is conservative for total energy and total number of particles. However, the SSP-RK3 scheme by Shu and Osher \cite{shu1988efficient} that we use for discrete time advancement is not time reversible and thus in general (when the Hamiltonian changes in time) decays total energy on the order of the time-stepping scheme (demonstrated in Table 2. of \cite{hakim2019discontinuous}). 
However, in the cases considered here where the Hamiltonian is time-independent, energy is conserved exactly---see Section~\ref{sec:Benchmarks} and the energy conservation errors in Table~\ref{table:energy_cons_for_shocks}.

Further, SSP-RK3 provides intrinsic stability advantages over lower order explicit integrators as SSP-RK3 can safely model standing modes without the need for additional dissipation due to its absolute-stability region containing the origin including a finite range of imaginary frequencies in the complex plane.\cite{durran2010numerical, tadmor2025stability} This restriction does not exist for the collision terms however, and that operator may be treated separately.

One such option for the operators on the right hand side of the kinetic equation, such as BGK collision term \ref{eq:BGK_op}, 
\begin{align}
    \frac{\partial f}{\partial t} + \{ f, H\} =  -\nu(f - f_M),
\end{align}
is to also update this term explicitly with the same scheme. However, in the limit that $\nu \rightarrow \infty$, this term puts extreme restrictions on the time step, which tend towards $\Delta t \rightarrow 0$. To avoid the collisional time step restriction, we split the collision and advective terms. Evolving the advective term with an SSP-RK3
\begin{align}
  \pfrac{f}{t} + \{f,H\} = 0.
\end{align}
and then applying the BGK operator subsequently,
\begin{align}
  \pfrac{f}{t} = -\nu(f - f_M).\label{eqn:split_bgk}
\end{align}

SSP-RK3 advance in the advective equation takes $f^n$ initially and returns $f^*$. Then, to advance $f^*$ to $f^{n+1}$, we use a separate time stepper for collisions and construct $f_{M}$ with identical discrete moments to $f^*$. Replacing the time partial derivative in equation \ref{eqn:split_bgk} with a first-order implicit Euler step, equation \ref{eqn:split_bgk} then becomes
\begin{align}
    f^{n+1} = \frac{f^{*} + \nu \Delta t f_M}{1 + \nu \Delta t} \label{eqn:implicit_bgk}
\end{align}
which has no time step restriction. This is unlike the explicit time step restriction for validity, $\nu \Delta t \ll 1$. In the low-collisionality limit where $\nu \Delta t \ll 1$, equation \ref{eqn:implicit_bgk} matches the explicit BGK operator up to first order as $f^{n+1} = f^* - \nu \Delta t(f^* - f_M)$. In the limit that the collision frequency is much larger than the timestep size, $\nu \Delta t \gg 1$, the implicit form, equation \ref{eqn:implicit_bgk}, has the correct asymptotic behavior, $f^{n+1} = f_M$, which is the fluid limit. Therefore, the collision operator is accurate for all choices of $\nu \Delta t$.\footnote{The time integration described here is a simple form of an implicit-explicit (IMEX) scheme; see, for example, \cite{pieraccini2007implicit}. High order SSP IMEX schemes have been developed (see \cite{pareschi2005implicit}) and are planned as future additions to the Gkeyll code. }

\begin{remark} The overall scheme is \textit{Asymptotic Preserving}; i.e., in the fluid limit, the scheme reproduces the ideal fluids (Euler) fluxes. Section \ref{sec:Benchmarks} demonstrates this property through agreement between exact Riemann solutions and the moments of the kinetic scheme. 
\end{remark}

% Adapted from the MJ Paper
\subsection{Iterative Corrections for Collisions}\label{subsec:bgk_iterative_corr}

For use in the BGK collision operator we need an equivalent Maxwellian, $f_M$. When this Maxwellian is constructed discretely on a finite grid it's moments will not, in general, match the moments used to construct the Maxwellian itself. If not accounted for this error in constructing the discrete Maxwellian can cause significant errors in the conservation properties of the full scheme. To fix this, we employ an iterative algorithm to correct the discretely projected $f_M$ to match the desired moments. 

The algorithm begins with the desired moments, $\sqrt{g}n_0$, $\mvec{u}_0$, and $T_0$ which are used to create a discrete projection $\tilde{f}_M$, where $\tilde{f}_M$ represents the uncorrected equilibrium distribution. These moments either come from the initial conditions or from the moments of $f$. Because of the discrete, finite velocity space grid and the finite polynomial order representation of the solution, the projected moments of the discrete distribution do not necessarily match the desired values. This error is ultimately what this iterative algorithm is designed to correct.

We can calculate the moments via Section \ref{sec:Equilibrium Distribution Function} of the discrete distribution to obtain the corresponding moments $\sqrt{g}\tilde{n}$, $\tilde{ \mvec{u}}$, and $\tilde{T}$. We then use Picard iteration to minimize the difference between the vector of desired moments $\mvec{M}_0 = (\sqrt{g}n_0, \mvec{u}_{0}, T_0)$ and the vector of moments from the discrete projection $\mvec{M}_k = (\sqrt{g} \tilde{n}, \tilde{\mvec{u}}, \tilde{T})$. For an iteration index $k$, the initial step, $k = 1$ and with $d\mvec{M}_{1} = 0$, the iterative scheme to correct the moments is

\begin{align} \label{iterative_correction}
\begin{split}
    & dd\mvec{M}_k = \mvec{M}_0 - \mvec{M}_k  \\
    & d\mvec{M}_{k+1} = d\mvec{M}_k + dd\mvec{M}_k  \\
    & \mvec{M}_{k+1} = \mvec{M}_0 + d\mvec{M}_{k+1}.
\end{split}
\end{align}
Once we take a step from $k$ to $k+1$, the moments $\mvec{M}_{k+1}$ are then used to re-project the distribution function $\tilde{f}_M$; $k$ is incremented, and $\mvec{M}_k$ is set equal to the new moments of $\tilde{f}_M$. This repeats until the moments converge to a desired error. Once the correction algorithm is converged, the Maxwellian distribution is ready to be used inside the BGK collision operator, ensuring local conservation.

It is important to note that the velocity bounds need to be sufficient for this iteration scheme to converge. For instance, the bounds need to be wide enough to resolve the temperature and beam velocity. For most adequately resolved cases, by which we mean the moments are representable, this algorithm converges to machine precision in roughly 3 to 15 iterations. Distributions with moments near the edge of the resolvability range tend to slow down the iterator considerably.\footnote{See \cite{johnson2025moment} for a special relativistic analysis that applies to the non-relativistic case here.} Once complete, the Picard iteration returns a discretely projected $f_M$ that is moment matched to the desired values. Having this algorithm is essential for maintaining both conservation laws and the asymptotic preservation property of the discrete system. In particular, this correction algorithm is what permits the use of the implicit scheme defined in equation \ref{eqn:implicit_bgk}; also see \cite{liu2025asdex} which uses this algorithm for multi-species collisions.\footnote{ Alternative discrete equilibrium correction algorithms have been developed, such as \cite{mieussens2000discrete} and \cite{dzanic2023positivity}, which use Newton corrections that typically take only a few iterations to reach machine precision accuracy. Such an approach could conceivably be applied here. However, specifically for the modal DG represented quantities, the Jacobian inversion required would be of the size of the number of basis functions in configuration space times the number of moments, squared, which makes this approach less attractive. The trade off is that the corrective method provided here will take additional iterations but is simpler to implement and can automatically handle arbitrary equilibrium functions, such as the  Maxwellian, Maxwell-Jüttner, and equilibrium in curved spaces.}

\subsection{Discrete Energy Conservation}
In addition to the iterative scheme for collisions, the moments themselves must be computed consistently. To compute the equilibrium, $f_M$ in equation \ref{eq:lab_frame_eq}, the density and momentum density moments are given by equations \ref{eq:Jn} and \ref{eq:Jnu_i} discretely. However, the temperature moment, equation \ref{eq:v_th^2}, must be expanded carefully in terms of quantities that belong to the discrete basis. A consistent choice where all terms reside within the reduced Serendipity basis is,
\begin{align}
    dnv_{th}^2\sqrt{g} &= 2E - \sqrt{g}ng^{ij}u_iu_j \label{eqn:proper_comp_of_vth^2}\\
    E &\equiv \int H_hfd\mvec{p}.
\end{align}
An alternative choice, which is seemingly equivalent, is to replace $2E$ with $h^{ij}M_{2ij}$ where $M_{2ij} \equiv \int p_i p_j f d\mvec{p}$. However, this choice is not energy conserving. This is due to the computation of $M_{2ij}$, which has terms ~$x^2p^2$ that do not belong to the Serendipity basis set. It is therefore important for the energy conservation of the collision algorithm that we use equation \ref{eqn:proper_comp_of_vth^2} to compute temperature, as $H_h$ does belong to the basis set.

\subsection{Coordinate Singularities}
Coordinate singularities may exist for geometries such as polar or spherical coordinates as the radius tends to zero. For the purposes of this paper, these regions are excluded from the simulation domain. In the relativistic context, black hole representations have both physical and coordinate singularities that may arise at horizons in certain choices of coordinates, like Boyer-Lindquist. These coordinate singularities, however, are removable with horizon-penetrating coordinates such as Kerr-Schild. This leaves only the ring-singularity at $r = 0$, which may also be excluded from the domain for most simulations which are interested in dynamics outside of the horizon. 

As introduced here, the solver would require additional boundary conditions to handle coordinate singularities like $r=0$ in cylindrical and spherical coordinates. The CFL condition may also become dominated by these regions. This could be potentially remedied by grid-coarsening near the coordinate singularities. An example of the treatment of a singularity in a DG scheme is the x-point in 2D gyrokinetic simulations, see \cite{shukla2025constructing}, which use nodal calculations at surface quadrature points to avoid directly evaluating the geometry at the singularity. For all the simulations considered in this manuscript, coordinates singularities are excised from the domain, similar to the excision of singularities in general relativistic simulations of compact objects.
% -*- latex -*-

\section{Benchmarks and Example Problems}\label{sec:Benchmarks}

Reference quantities, such as ones for density, momenta, temperature, length scale, angular speed, and time are set to unity. Therefore, these benchmarks are dimensionless and only intended to test the implementation of the DG scheme. In addition, all tests in this section were run with the $\texttt{Gkeyll}$ code using polynomial order $p = 2$ Serendipity basis. The scripts are publicly available at: \url{https://github.com/ammarhakim/gkyl-paper-inp}

%% Additions by GJ:
\subsection{Sodshocks in Flat and Curved Geometries}
As a first test of the solver, consider a sodshock on a flat-space geometry with one configuration-space dimension and one momentum-space dimension (1x1v). Sodshock test cases provide an integrated test for initializing the moments, evolving the distribution with the canonical Poisson bracket updater, and testing the conservative collision operator. The shock is set up by initializing Maxwellians with a left state, $S_L$, for density and temperature moments, $S_L = [1.0, 1.0]$ for $x \leq 0.5$, and a right state $S_R = [0.125, \sqrt{1.0/1.25}]$ for $x > 0.5$, with no initial flows. The moments for the states are defined for density, covariant momentum components, and temperature moments as
\begin{align}
    n  &= \frac{1}{\sqrt{g}}\int_{-\infty}^{\infty} f d\mvec{p}\\
     u_i &= \frac{g_{ij}}{n \sqrt{g}}\int_{-\infty}^{\infty} \pfrac{H_h}{p_j} f d\mvec{p} \\
    T &=v_{th}^2 = \frac{1}{d n \sqrt{g}} \bigg( 2 \int_{-\infty}^{\infty} H_h f d\mvec{p} - n \sqrt{g} g^{ij}u_i u_j \bigg )
\end{align}
where the temperature moment is defined as the isolated $v_{th}^2 \equiv T$. These moments are plugged into $f_M$ equilibrium, equation \ref{eq:lab_frame_eq}, which is then run through the moment correction algorithm to project the initial distribution on the grid. 

Several simulations are taken with varied collision frequency between effectively collisionless, $\nu = 15$, and the fluid limit with a frequency, $\nu = 15000$. Reflecting-wall boundary conditions are located at $x = 0$ and $x = 1$, and the momentum bounds span from $p_{max} = \pm 8$ to capture the entire distribution function. A grid resolution of $N_x = 128$ and $NV_x$ = $32$ is used for these simulations, which run from time $t = 0$ until a final time of $t_f = 0.1$. A polynomial order of $p = 2$ is chosen with Serendipity basis. In this case, the Hamiltonian takes the simple form: $H = \frac{1}{2}p_x^2$.

In Figure \ref{fig:can_pb_sodshock}, the integrated moments of the kinetic sodshock simulations are compared to an exact Riemann solution for a range of collision frequencies. Beginning with weak collisions, $\nu = 15$, the free streaming particles advect away the initial jump in the moments, smoothing out the discontinuity. However, as the collisionality increases to $\nu = 15000$, the kinetic simulation enters the fluid limit, and features such as the rarefaction wave, contact discontinuity, and forward-propagating shock appear. The convergence in the highly collisional limit to the exact Riemann solution of the kinetic solver verifies that the scheme is asymptotic preserving. In other words, in the limit of high collisionality, the solver reproduces the Euler-fluxes and hence behaves like a fluid. This behavior is evident from the agreement between the exact solution, the black dashed line, and the yellow line, which represents the integrated moments of the distribution function for the $\nu = 15000$ case in Figure \ref{fig:can_pb_sodshock}.

\begin{figure}[H]
\centering
\includegraphics[width = 1.0\linewidth]{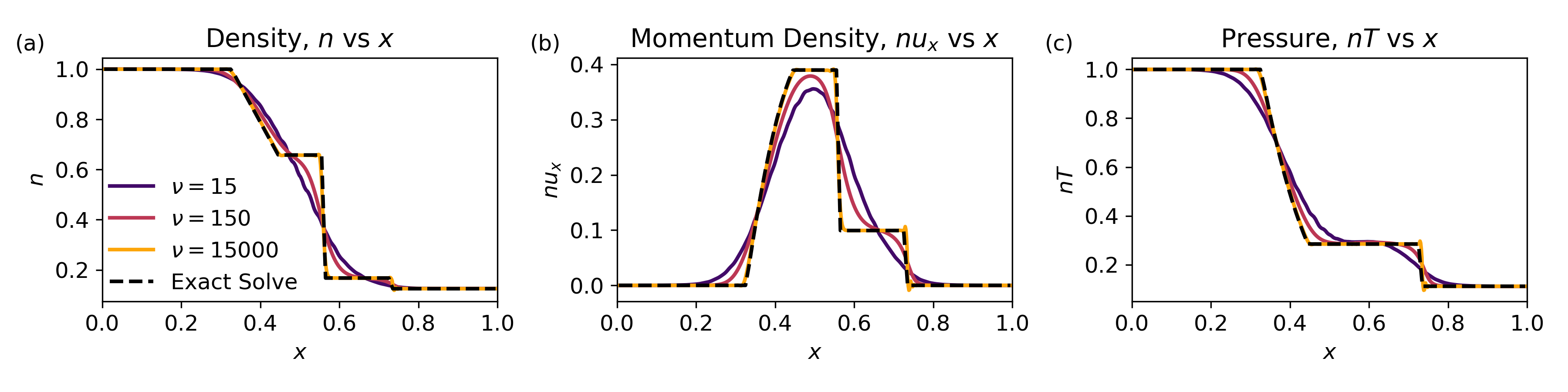}
\caption{Sodshock (1x1v) moment comparison taken at the final timestep between the moments of the canonical bracket solver at various collision frequencies, $\nu$, against the Riemann solution. The collision frequency spans from a nearly collisionless (free-streaming) limit to a highly collisional (fluid-like) limit. }\label{fig:can_pb_sodshock}
\end{figure}

The Rankine-Hugoniot conditions, while themselves not directly present in the kinetic solver, are recovered in the highly collisional limit. The moment conserving operation for projecting the distribution and the implicit solve for the collisions in this case rapidly relaxes $f$ towards the moment-matched $f_M$ locally in configuration space. In this asymptotic limit, the kinetic solver matches the fluid equations, whose weak forms obey the Rankine-Hugoniot conditions. Therefore, the jump conditions are recovered through the limiting fluid system rather than through direct enforcement inside the kinetic update.

Extending the sodshock problem to the non-rectangular geometry of an annular disk, consider a radially expanding shock with the same initial state as the 1D flat geometry case. The Hamiltonian for free streaming motion on an annulus is
\begin{equation}
    H = \frac{1}{2}\left(p_r^2 + \frac{p_\theta^2}{r^2} \right),
\end{equation}
where the metric for an annulus is given by,
\begin{equation}
    ds_{annulus}^2 = dr^2 + r^2 d\theta^2.
\end{equation} 
We use reflecting-wall boundary conditions at $r_{min} = 0.5$ and $r_{max} = 1.5$, a periodic grid azimuthally at $\theta_{min} = 0$ and $\theta_{max} = 2\pi$, and a grid size of $N_r = 128$, $N_\theta = 1$, and $NV_r = NV_\theta = 12$. A polynomial order of $p = 2$ is chosen with Serendipity basis. The left state transitions to the right states at $r_{jump} = 1.0$. We note that we have shifted the grid away from the origin because of the coordinate singularity at $r=0$.

Figure \ref{fig:can_pb_sodshock_annular_disk} presents a comparison of the moments at the final time, $t_f = 0.1$, between the integrated moments of the kinetic distribution and a high-resolution Euler solution using the same initial conditions. Like the 1x1v Sodshock case in figure \ref{fig:can_pb_sodshock}, figure \ref{fig:can_pb_sodshock_annular_disk} line-outs (panels d-f) demonstrate quantitative agreement between the canonical Poisson solver and the Euler solver. The mean relative-error over the domain in the density and temperature moments are: $\epsilon(n) = 3.967\times10^{-3}$ and $\epsilon(T) = 5.480 \times 10^{-3}$. The agreement between the Euler solution and kinetic solve for the annulus verifies that the asymptotic preserving property carries over to non-rectangular geometries.

To examine the convergence between the Euler and kinetic moments more closely, figure \ref{fig:annular_disk_convergence} zooms in on the forward-propagating shock at $r = 1.2$. As we increase spatial resolution, the moments of the distribution function converge to the Euler solution without increasing the magnitude of overshoots around the shock. We emphasize that these kinetic simulations are run without the limiters typically employed by discretizations of fluid equations; the $L_2$ stability of the canonical bracket scheme leads to this robust performance and stable convergence without any growth of spurious oscillations around the shock.

\begin{figure}[H]
\centering
\includegraphics[width = 1.0\linewidth]{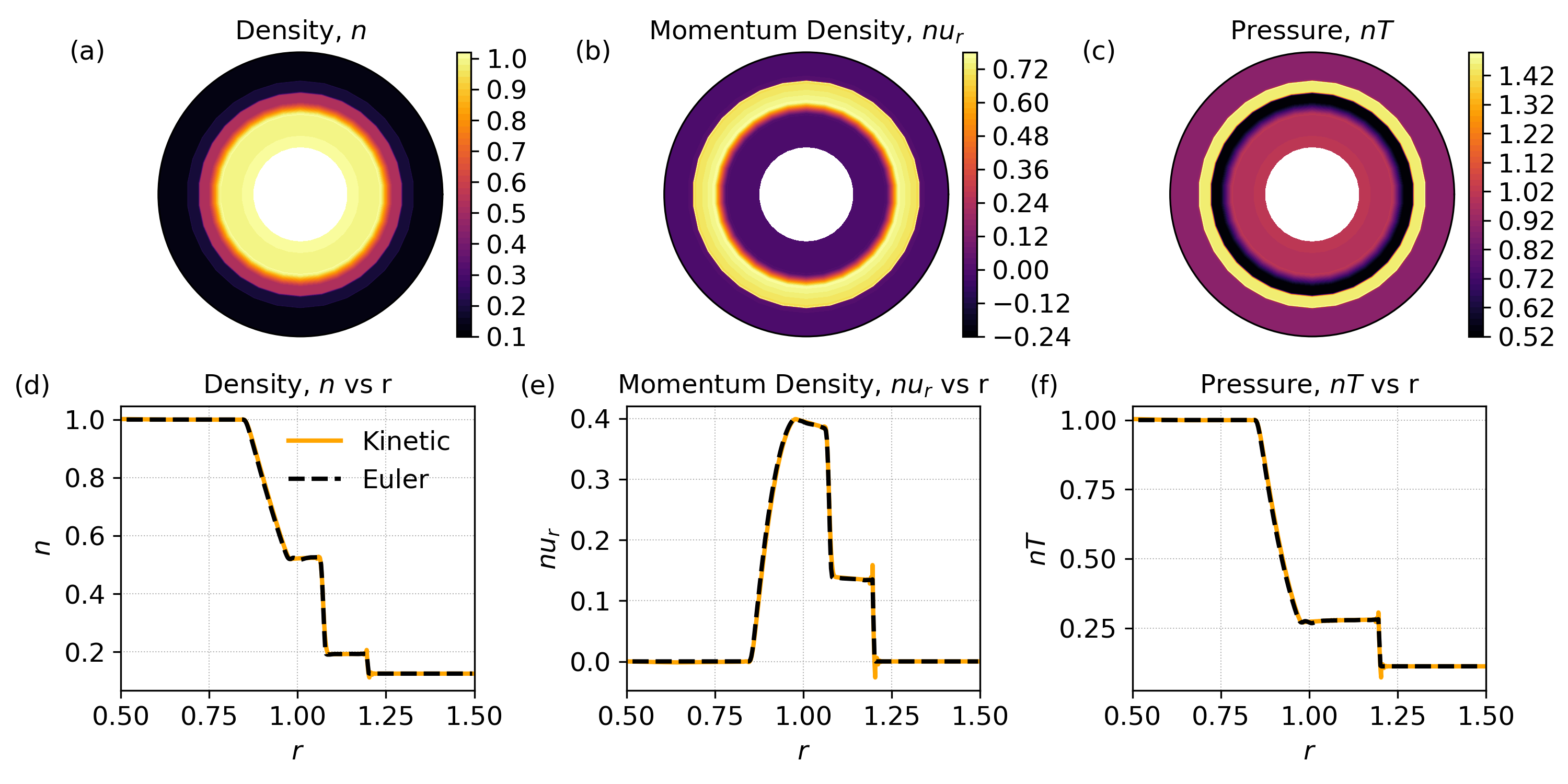}
\caption{Annular disk sodshock moments at the final simulation time using the canonical Poisson bracket formalism. From left to right, the columns show density, radial momentum density, and pressure. The top row shows the physical coordinates of the moments of the distribution, and the bottom row of plots show line-outs of the moments compared to the Euler solution.}  \label{fig:can_pb_sodshock_annular_disk}
\end{figure}

\begin{figure}[H]
\centering
\includegraphics[width = 1.0\linewidth]{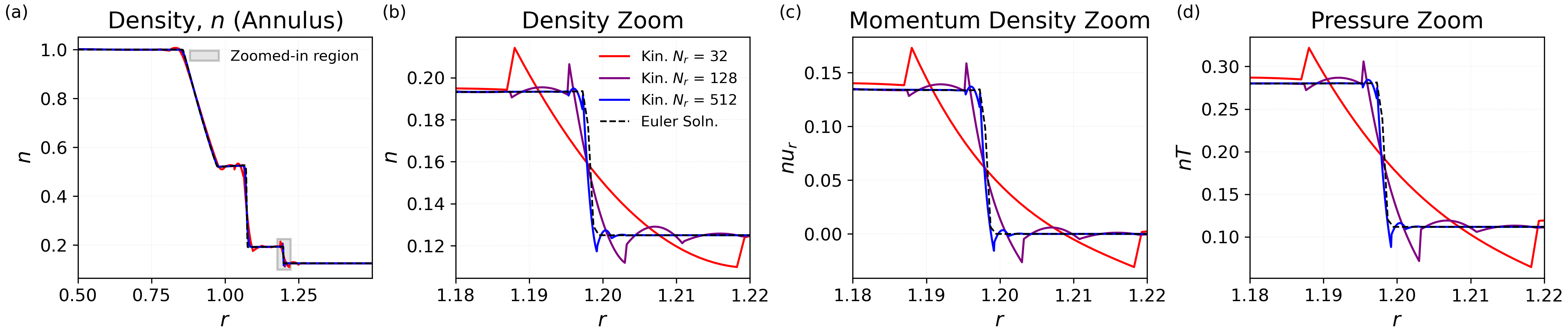}
\caption{ Zooming in onto the forward-propagating shock in the annulus sodshock configuration from figure \ref{fig:can_pb_sodshock_annular_disk} to examine convergence with increasing spatial resolution. Panel (a) has a gray highlighted region that is zoomed into in panels (b-d) for the density, velocity, and temperature moments respectively. }  \label{fig:annular_disk_convergence}
\end{figure}

In Table \ref{table:energy_cons_for_shocks}, we report the absolute error in the total (integrated) density and total energy for both sodshock cases. In each case, both total number of particles and total energy are verified to be conserved for the discrete scheme, demonstrating the proper implementation of our conservative scheme and the empirical verification of our proofs in Section~\ref{sec:scheme}.

\begin{table}[h!]
\centering 
\begin{tabular}{|c|c|c|c|c|}
\hline
& \multicolumn{2}{|c|}{1x1v Sodshock} & \multicolumn{2}{|c|}{Annular Sodshock} \\
\hline
& P1 & P2 & P1 & P2 \\
\hline
$\Delta N/N_0$ & $7.30 \times 10^{-15}$ & $1.91 \times 10^{-14}$ & $1.89 \times 10^{-14}$ & $3.42 \times 10^{-14}$ \\
\hline
$\Delta E/E_0$ & $2.60 \times 10^{-15}$ & $1.76 \times 10^{-14}$ & $2.54 \times 10^{-15}$ & $3.23 \times 10^{-14}$ \\
\hline 
\end{tabular}
\caption{The difference in total number of particles $N$, and total energy $E$ for polyorder 1 (P1) and 2 (P2) representations divided by their respective initial totals. The 1x1v sodshock data is from the specific $\nu = 15000$ case, but for any collision frequency both quantities are conserved machine accuracy. }\label{table:energy_cons_for_shocks}
\end{table}

Further, the annular disk case in Figure \ref{fig:can_pb_sodshock_annular_disk} conserves the angular momentum component $p_\theta$ to machine precision from the initial to final time of the simulation. This conservation of the azimuthal canonical momentum is only possible due to exact integration of the surface and volume terms---we have entirely eliminated aliasing errors from the solution \cite{hakim2020alias, juno2018discontinuous}. Without exact integration, the aliasing errors would pollute the invariant. Checking conserved quantities such as the preservation of $p_\theta$ in this axisymmetric case serves as a delicate and robust benchmark for verifying that the numerical scheme remains alias-free.

\subsection{Kelvin–Helmholtz Instability on Embedded Surfaces}
Consider next a 2-dimensional configuration-space and 2-dimensional momentum-space (2x2v) setup: the surface of a sphere. The aim is to reproduce the Kelvin-Helmholtz instability with the kinetic solver in the highly collisional limit, $\nu = 15000$. This problem is set up initially with counter-propagating fluid elements that are unstable to the Kelvin-Helmholtz instability (KHI). Random modes are initialized in both momentum moments to onset the instability with a perturbation amplitude of $1 \times 10^{-2}$ and modes spanning from the system size down the the grid scales. In this configuration, the initial distribution function moments with density, $\theta$-momentum, $\phi$-momentum, and temperature are given by the state $S_1 = [2.0, 0.0, -0.5, 1.25]$, which is initialized around the equator for all $\phi$, and between $\theta = (\pi/4, 3\pi/4)$. The second element, propagating in the opposite direction, has the state $S_2 = [1.0, 0.0, 0.5, 2.5]$ on the remaining parts of the domain, up to a reflective boundary at $\theta = \pi/8$ and $7\pi/8$. The simulation resolution is $N_\theta = N_\phi = 192$ and $NV_\theta = NV_\phi = 12$. A polynomial order of $p = 2$ is chosen with Serendipity basis. The flows come from Hamiltonian constructed using the metric for the surface of a sphere, 
\begin{equation}
    ds_{sphere}^2 = r^2 d\theta^2 + r^2 \sin^2\theta d\phi^2.
\end{equation}

The density snapshot on the surface of the sphere in Figure \ref{fig:khi_sphere} is taken at time $t=3.0$, after the KHI vortices have begun to form. Representative points on the distribution have been selected with a red ``$x$'' and the corresponding local $df = f - f_{Eq}$ are plotted to the right, where $f_{Eq}$ is the moment-matched local thermodynamic equilibrium. From this quantity, moment-closures for reduced models can be numerically explored, and transport coefficients for momentum and energy can be computed. The quantity $df$ is thus a crucial piece of information provided by a kinetic approach with wide utility for understanding the transport of these physical systems. 

We emphasize that a benefit of the DG kinetic continuum representation is the accuracy of this calculation of $df$, which is here on the order of $10^{-5}$ relative to the peak magnitude of the distribution. Indeed, both $f$ and $f_{Eq}$ are $\mathcal{O}(1)$ quantities around their peak. The power of a continuum representation of the distribution function compared to Monte Carlo approaches is thus apparent due to the sensitivity of the analysis we can perform with our approach. 

\begin{figure}[H]
\centering
\includegraphics[width = 1.0\linewidth]{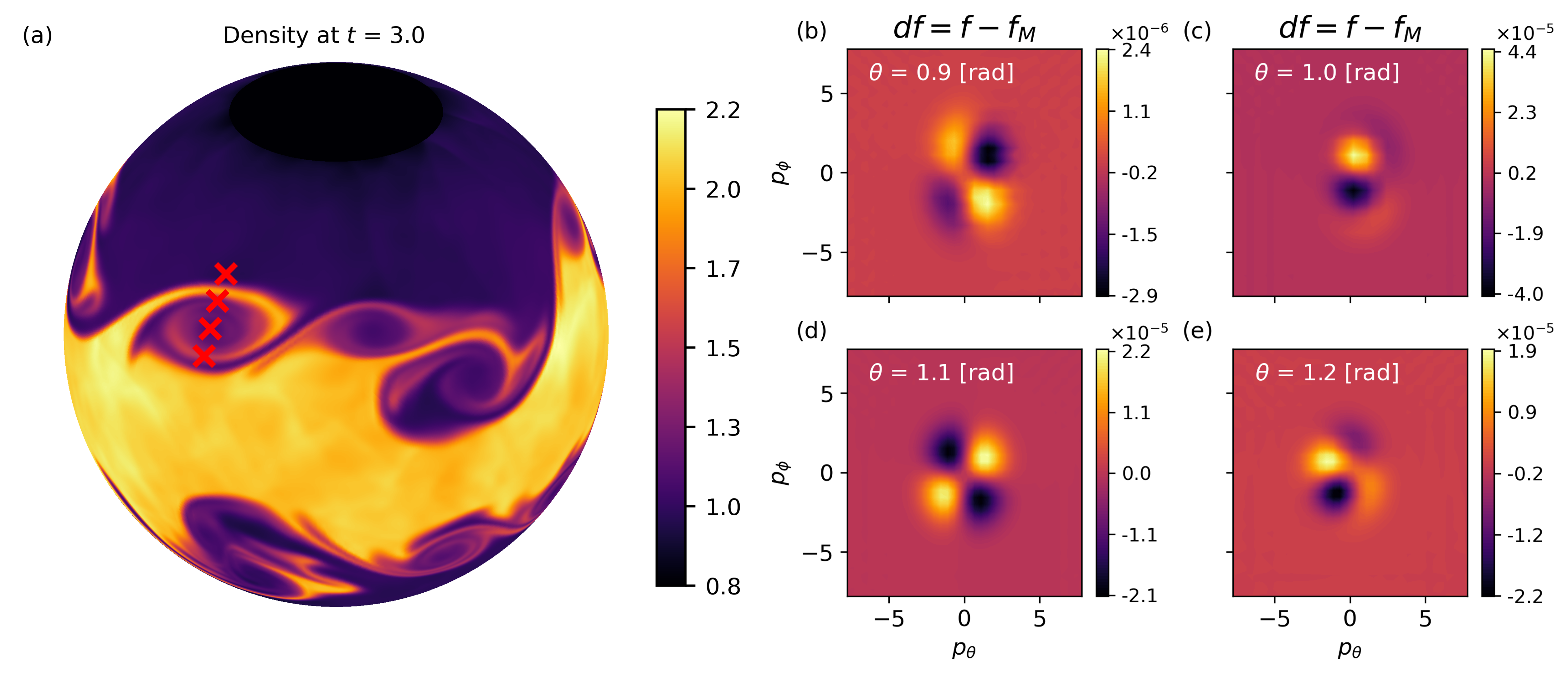}
\caption{Kelvin-Helmholtz instability of the surface of a sphere using the canonical Poisson bracket method to update the distribution function, $f$. Panel (a) is the density moment of the distribution in physical space. Panels (b-e) show the difference between the distribution and local moment-matched Maxwellian, called $df$, at select points indicated by the red-x's on panel (a).}  \label{fig:khi_sphere}
\end{figure}

We next illustrate the Kelvin-Helmholtz instability on a hyperboloid surface. We have an angular coordinate, $\theta$, which spans from 0 to $2\pi$. The second coordinate $z$ spans from $(-1, 1)$. The initial distribution function's moments in the hyperbolic configuration with density, $\theta$-momentum, $z$-momentum, and temperature are given by the state $S_1 = [2.0, 0.0, -0.5, 1.25]$, which is initialized around the equator for all $\theta$ with between $|z| \leq 0.5$. The second element, propagating in the opposite direction, has the state $S_2 = [1.0, 0.0, 0.5, 2.5]$ on the remaining parts of the domain up to a reflective boundary at $z = \pm 1$. The number of cells are $N_\theta = N_z = 192$ and $NV_\theta = NV_z = 12$. A polynomial order of $p = 2$ is chosen with Serendipity basis. The metric for the hyperbolic geometry case is given by,
\begin{equation}
    ds_{hyperbolic}^2 = (r^2 +z^2) d\theta^2 + \frac{r^2 + 2z^2}{r^2 + z^2} dz^2.
\end{equation}
The factor $r$ is fixed at a $r = 0.5$, which dictates the thinnest radial point such that a cross-section at $z = 0$ is a circle with this radius.

Like the surface of a sphere, we also display here in figure \ref{fig:khi_hyperbolic_surf} the physical geometry for density and the computation grid moments at time $ t= 4.5$. As with the surface of a sphere geometry, the density displays the familiar cat's eye structure of the nonlinear state of the KHI, albeit with altered shear due to the difference in geometry. The purpose of this figure is, thus, to illustrate that the algorithm can handle various specifications for the geometry by simply changing the metric used in the Hamiltonian. The specific implementation of this solver in the \texttt{Gkeyll} code can handle any coordinate configurations from (1x1v) to (3x3v).

\begin{figure}[H]
\centering
\includegraphics[width = 0.95\linewidth]{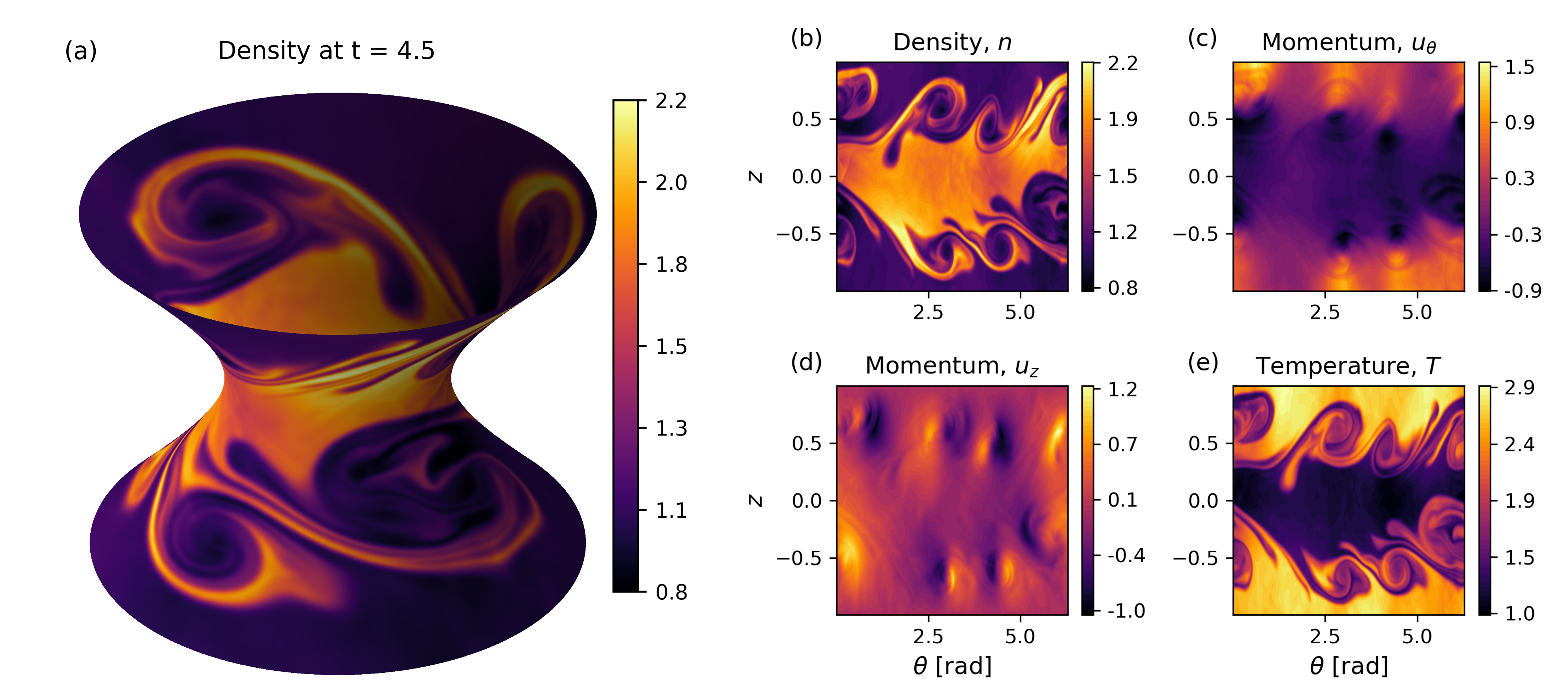}
\caption{Kelvin-Helmholtz instability on a hyperbolic surface. To the left, the density moment is mapped onto the physical surface. The panels to the right show the moments on the surface of density, $\theta$ and $z$ momenta, and temperature. The moments in all panels are snapshots taken at time, $t = 4.5$.}  \label{fig:khi_hyperbolic_surf}
\end{figure}

\subsection{Verification of $L_2$ Monotonic Decay}
As a check of the properties of the discrete scheme, we plot the $L_2$ norm for the Kelvin-Helmholtz instability in spherical and hyperbolic geometries. We compute the integral for the $L_2$ norm over the entire domain of $f$ at each timestep. Plotting over the entire simulation duration, the result shown in Figure \ref{fig:L2_norm}, verifies the monotonic $L_2$ decay property for both cases.

\begin{figure}[H]
\centering
\includegraphics[width = 0.45\linewidth]{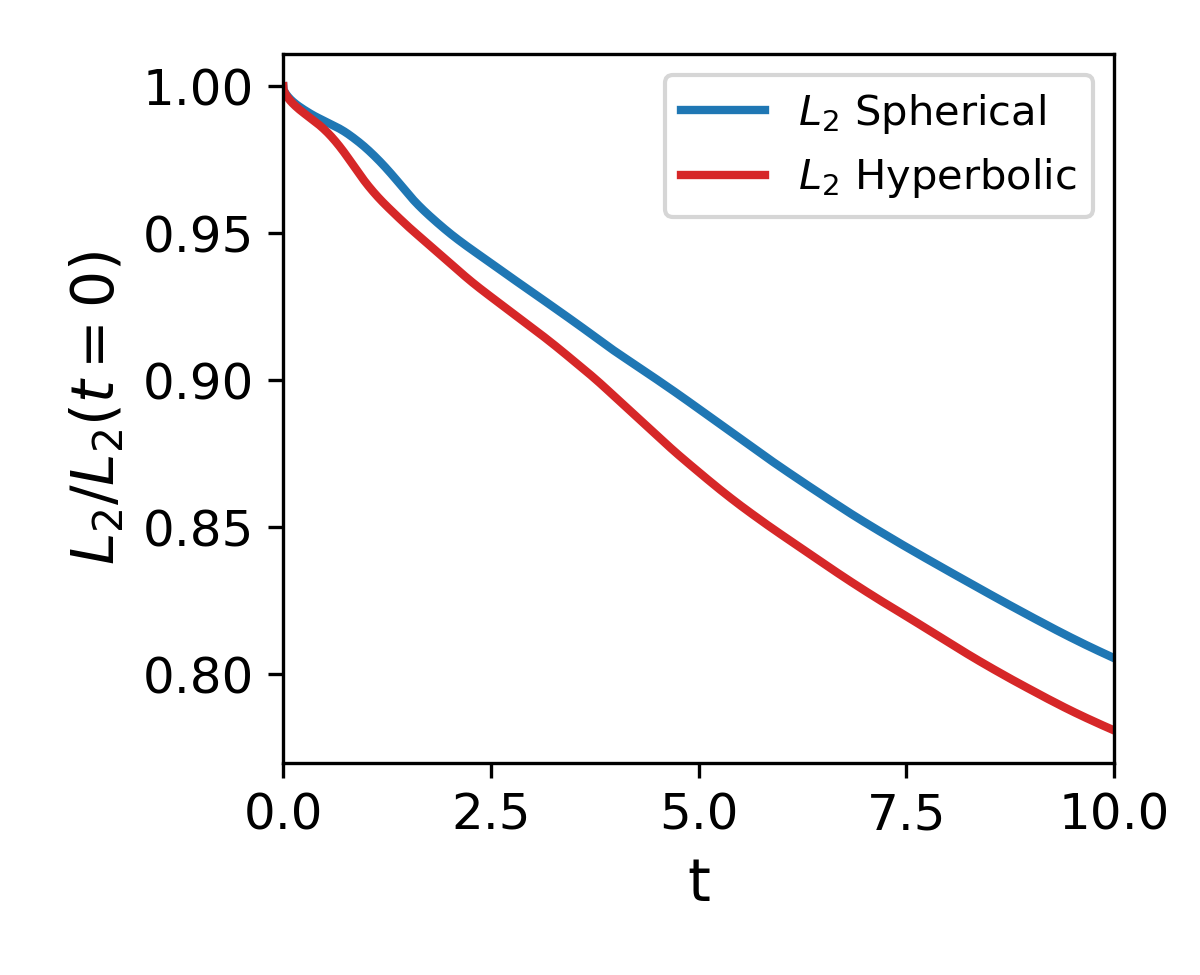}
\caption{$L_2$-norm for both the spherical and hyperbolic KHI normalized to the initial value of the $L_2$ norm.}  \label{fig:L2_norm}
\end{figure}

\subsection{Motion on a Uniformly Rotating Sphere}

As a final test, consider a uniformly rotating sphere. Here, we can compare the deflected trajectory of an initial Gaussian bump with an angular velocity directed towards the pole and an initial azimuthal angular speed both equal to that of the uniformly rotating sphere's angular speed, $\Omega = 1.0$. There are no collisions as we aim to compare this result to the single particle trajectory. 

The initial distribution for $f(t=0)$ is given by
\begin{align}
    f_0 
    = 
    \frac{n}{2\pi T}
    \exp \left( - \frac{(\theta-\theta_0)^2 + (\phi- \phi_0)^2 \sin^2 \theta}{2\sigma_{bump}}  \right)
    \exp \left( - \frac{1}{2T} \left((p_\theta - p_{\theta,0})^2 + \frac{(p_\phi - p_{\phi,0})^2}{ \sin^2 \theta}  \right) \right)
\end{align}
where a low temperature of $T = 1/6400$ minimizes the thermal spread relative to the width of the bump, $\sigma_{bump} = 0.1$. The peak is localized around $[\theta_0, \phi_0, p_{\theta,0}, p_{\phi,0}] = [\pi/2, \pi, 1, 1]$. The domain only has the upper hemisphere between $\theta = (\pi/2, 7 \pi/8)$. The number of cells in each direction of the phase space is given by $N_\theta$ = 24, $N_\phi = 172$, $NV_\theta = 192$, and $NV_\phi = 24$ with bounds on the momentum of $p_\theta = (-1.0625, 1.0625)$ and $p_\phi = (0.90625, 1.0625)$. These choices for bounds are specifically chosen to resolve the thermal width around the conserved variable $p_\phi$ and the range of momenta attained by the particle in $p_\theta$ through the particles' arcs. Finally, a polynomial order of $p = 2$ is chosen with Serendipity basis.

A single particle trajectory for these exact initial conditions is compared in both panels of Fig.\thinspace\ref{fig:coriolis} to indicate the trajectory of the bump. It can be seen from both panels that the collection of particles agrees nicely without sharp oscillations from the initial discontinuity at the reflective bound at $\theta = \theta_0$. Ultimately, the stars that indicate the exact particle position are expected to trail the bump, as they do for all frames.

Further, the reorientation of the shape to first compress in the $\theta$-direction and then invert is the expected behavior for the extended shape. Individually, the test particle undergoes initial deflection due to the Coriolis force at frame 1. By frame 2, the effective potential begins altering the trajectory by reflecting the particle back towards the equator. At frames 3 and 4, the particle has been reflected by the effective potential, and the Coriolis force begins to act in the other direction, steepening the trajectory. This demonstrates agreement between geodesic motion and the collisionless canonical bracket solver.

\begin{figure}[H]
\centering
\includegraphics[width = 1.0\linewidth]{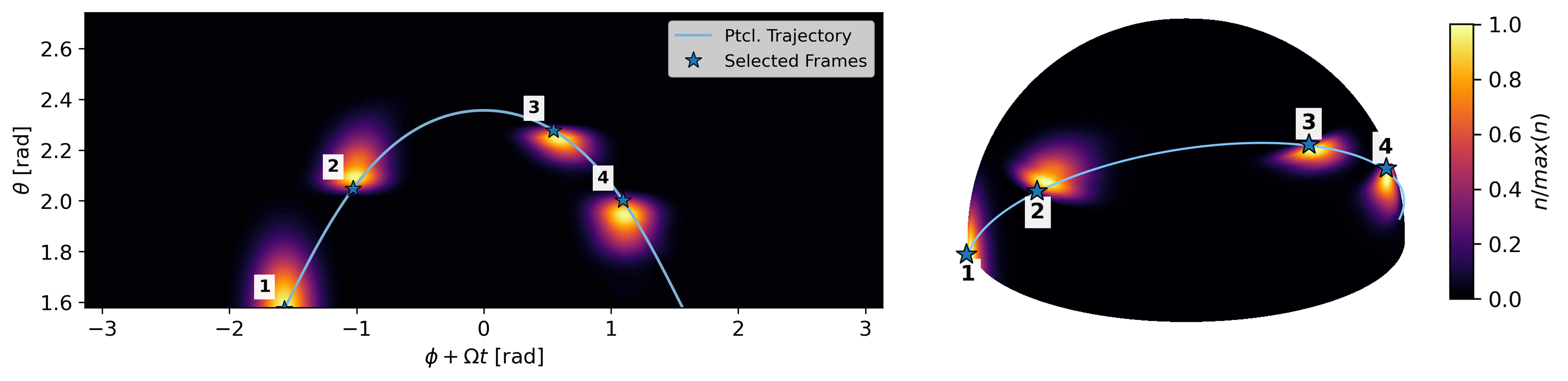}
\caption{Density plots at representative times of a blob moving on a uniformly rotating sphere. Both panels show the same data where the left is in computational coordinates and the right is the projection onto physical coordinates. Each label $(1...4)$, represents the location of the blob at times $t = (0.025, 0.500, 1.400, 1.775)$. At each point in time, the density is normalized to $\max(n)$. }  \label{fig:coriolis}
\end{figure}
% -*- latex -*-
\section{Normalized Momentum Coordinates}\label{sec:norm_momenta}

An intrinsic property of the canonical bracket formulation is that momentum space is expressed in terms of dual vectors that compress or stretch with the geometry of the system. As long as there is sufficient resolution to resolve the distribution function at all points in momentum space, this representation does not affect the solution. While all finite grids have minimum and maximum velocities and temperatures of the distribution that they are able to resolve, the momentum's dependence on geometry changes the resolvable moment ranges from point to point in the spatial domain. It should be noted that the discrete conservation laws proven earlier still hold even when $f$ enters under-resolved regions, but in these cases, the distribution can be distorted and lead to discrepancies with the exact solution.

Here we demonstrate a case with an annular disk where resolution in momentum space limits the resolvable distribution near small radii. To address this, we have developed a non-canonical bracket that uses normalized coordinates to represent the momentum vectors, avoiding the geometric dependence of momentum space resolution. 

Initializing uniform density and temperature with no flows in an annular disk geometry between a radius of $r = 0.5$ and $1.5$ in Figure \ref{fig:shrinking_vel_space_diagram} illustrates the momentum space dependence on geometry. Panels (a) and (c) specifically show how the uniform temperature and density for a Maxwellian on a static grid can take two separate ranges in $p_\theta$ depending on their radial locations of $r = 0.5$ and $1.5$, respectively. This radial dependence arises because the annular disk dual vectors in the $\theta$ direction have a factor of $1/r$,
\begin{align}\label{eqn:tangent_vectors_annulus}
    \dbasis{r} 
    &=
    \cos \theta\cbas{1}
    +
    \sin \theta\cbas{2}\\
    \dbasis{\theta} 
    &=
    -\frac{1}{r}\sin \theta\cbas{1}
    +
    \frac{1}{r}\cos \theta\cbas{2}.
\end{align}
In canonical coordinates, we describe the momentum vector as $\mvec{p} = p_i\dbasis{i}$. Therefore, the momentum $p_\theta$ in the $\theta$ coordinate will have a different distribution width for the same temperature, depending on the value of $r$. To avoid the challenge of finding static momentum bounds that can resolve the required temperature range, we can normalize the momentum coordinates. This, however, leads to a non-canonical Poisson bracket, which we derive next.

% Annular disk figure (Uniform T and N, divide out J)
\begin{figure}[H]
\centering
\includegraphics[width = 1.0\linewidth]{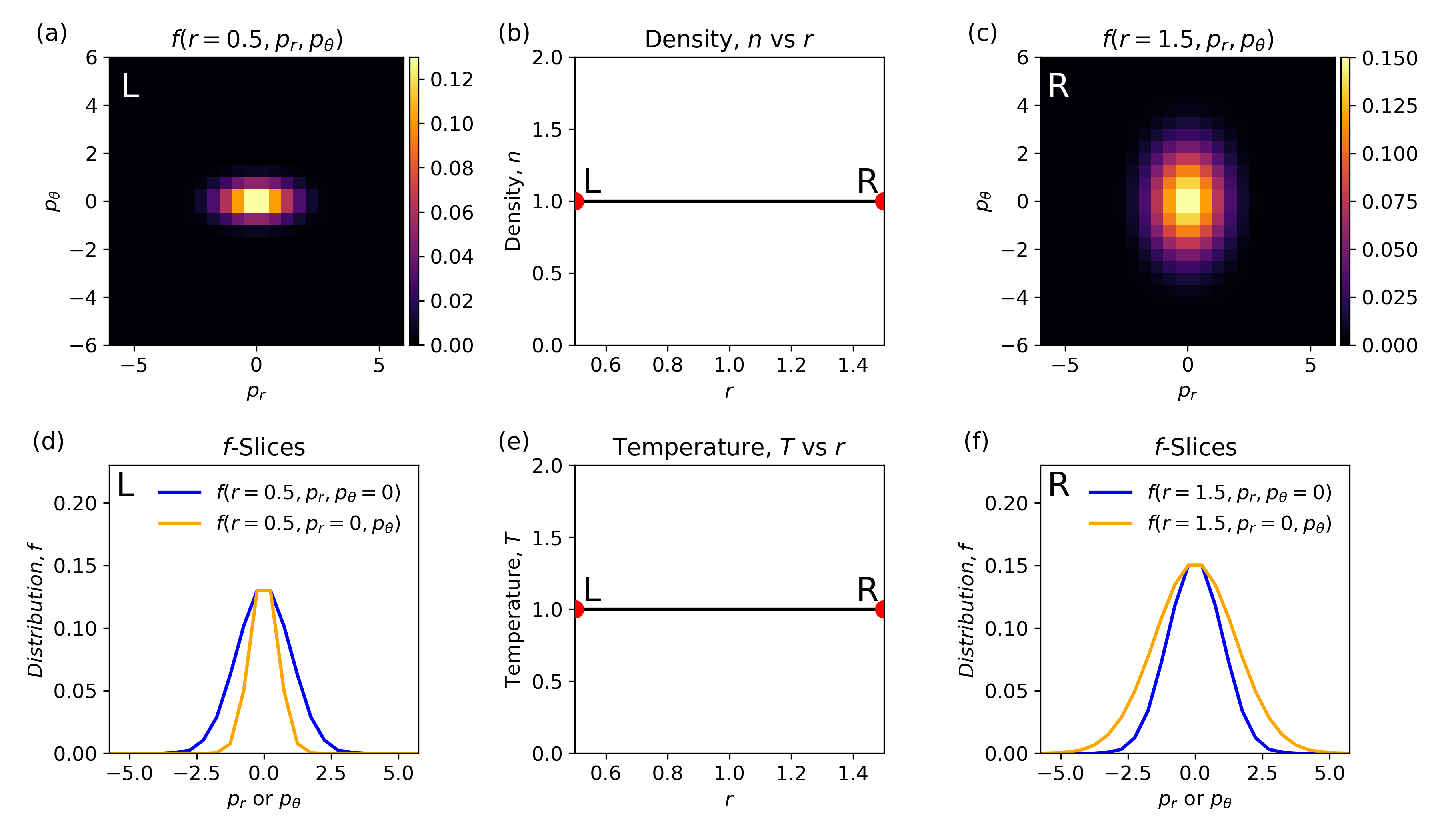}
\caption{Illustration of the velocity compression effect. The moments of density, panel (b) and temperature, panel (e), which are uniform from a radius of 0.5 to 1.5. There are red marks on the middle two panels to indicate the spatial locations of panel (a) and (d) at $r = 0.5$ to the left and panels (c) and (f) to the right at $r = 1.5$. Panels (a) and (c) plot the distribution function in momentum space at the left and right edges of the domain respectively. Below in panels (d) and (f) we take slices of the distribution function in each direction ($p_r$, $p_\theta$).}\label{fig:shrinking_vel_space_diagram}
\end{figure}

\subsection{Normalized Momenta Coordinates}\label{subsec:norm_coordinates}

Positionally dependent momentum space resolution can be completely obviated for diagonal metrics by introducing normalized momentum coordinates. Changing to normalized momentum coordinates leads to a change in the bracket from canonical to non-canonical form. Here we present a derivation of this non-canonical bracket with normalized momentum coordinates, denoted by hats, $\hat{p}_i$. We then demonstrate that this bracket also produces the equations of geodesic motion.

Beginning with the canonical form of the phase space Lagrangian and Hamiltonian
\begin{align}
    \mathcal{L} = \mvec{p} \cdot \dot{\mvec{x}} - H
\end{align}
\begin{align}
    H = \frac{1}{2}\mvec{p} \cdot \mvec{p}
\end{align}
we introduce normalizations to the momentum coordinates; for instance, we can expand the momentum in terms of tangent or normalized vectors $\mvec{p} = \sum_i p_i\dbasis{i} = \sum_i \hat{p}_i\ndbasis{i}$. (In this section, to avoid confusion, we write out all summations explictly and drop the Einstein summation convention). Here, the  tangent vectors are related to the normalized vectors via $\basis{i} = \vnorm{\basis{i}}\nbasis{i}$. In order to isolate the momentum components in terms of the normalized momentum coordinates, we introduce an associated dual basis $\ndbasis{i} \cdot \nbasis{j} = \delta\indices{^i_j}$, where we still have the original tangent and dual basis relation $\dbasis{i} \cdot \basis{j} = \delta\indices{^i_j}$. Note that normalization does not mean the bases are orthogonal. Introducing a new metric, $\eta$, for the normalized coordinates, we have 
\begin{align}
    \eta_{ij} \equiv \nbasis{i} \cdot \nbasis{j} = \frac{g_{ij}}{\vnorm{\basis{i}}\vnorm{\basis{j}}}. \label{eqn:eta_ij}
\end{align}
Furthermore, utilizing equation \ref{eqn:eta_ij} and the identities $g^{ij}g_{jk} = \eta^{ij}\eta_{jk} = \delta\indices{^i_k}$, we have a second equation for the contravariant components
\begin{align}
    \eta^{ij} = g^{ij}\vnorm{\basis{i}}\vnorm{\basis{j}}. \label{eqn:eta_ij_inv}
\end{align}
Here, it is important to note that there is no implied sum on repeated indices here. 

We can replace the momentum and time derivative of the position with expressions in terms of the normalized duals and tangent basis, respectively 
\begin{align}
    \mvec{p} &= \sum_k \hat{p}_k \ndbasis{k} \\
    \dot{\mvec{x}} &= \sum_k \dot{x}^k \basis{k}
\end{align}
which then gives expressions for replacing $\mvec{p}$ and $\dot{\mvec{x}}$ in the phase space Lagrangian:
\begin{align}
    \mathcal{L}(x^i,\dot{x}^i,\hat{p}_i,\dot{\hat{p}}_i) 
    &= 
    \sum_i \hat{p}_i\dot{x}^i\vnorm{\basis{i}}
    - 
    \sum_{i,j} \frac{1}{2}\eta^{ij}\hat{p}_i\hat{p}_j. \label{eqn:norm_momentum_phase_space_lagrangian}
\end{align}
Computing the Euler-Lagrange equations of $\mathcal{L}$, equation \ref{eqn:norm_momentum_phase_space_lagrangian}, we obtain evolution equations for $\dot{x}^k$ and $\dot{\hat{p}}_k$.
\begin{align}
    \dot{x}^k 
    &=
    \sum_{j} \frac{\eta^{kj}\hat{p}_j}{\vnorm{\basis{k}}}
    \label{eqn:norm_p_x_dot}\\
    \dot{\hat{p}}_k 
    &=
    \frac{1}{\vnorm{\basis{k}}}
    \bigg[ 
      \sum_i
      \bigg( 
        \hat{p}_i \pfrac{\vnorm{\basis{i}}}{x^k} 
        -\hat{p}_k \pfrac{\vnorm{\basis{k}}}{x^i}
      \bigg) 
      \dot{x}^i
      -
      \sum_{i,j}
      \frac{1}{2} \pfrac{\eta^{ij}}{x^k}\hat{p}_i\hat{p}_j
    \bigg] 
    \label{eqn:norm_p_p_dot}
\end{align}
We can relate these norms to the metric to simplify these expressions, $\vnorm{\basis{k}} = \sqrt{\basis{k} \cdot \basis{k}} = \sqrt{g_{kk}}$ (no sum implied). Using the equations of motion \ref{eqn:norm_p_x_dot} and \ref{eqn:norm_p_p_dot} and the partials of the Hamiltonian, we arrive at the bracket for the distribution $f(x^k,\hat{p}_k,t)$
\begin{align}
    \{ f, H\} 
    =
    &\sum_k \frac{1}{\sqrt{g_{kk}}} 
    \left(
    \pfrac{f}{x^k}\pfrac{H}{\hat{p}_k}
    - 
    \pfrac{H}{x^k}\pfrac{f}{\hat{p}_k}  
    \right) \notag \\
    &+
    \sum_{i,k}  \frac{1}{\sqrt{g_{kk} g_{ii}}} 
    \left(
    \hat{p}_i \pfrac{\left( \sqrt{g_{ii}}\right)}{x^k} 
    -
    \hat{p}_k \pfrac{\left( \sqrt{g_{kk}}\right)}{x^i}
    \right) 
    \pfrac{f}{\hat{p}_k}\pfrac{H}{\hat{p}_i}
\end{align}
with the Jacobian of the Poisson Tensor, $\mathcal{J} = 1/\sqrt{\det(\Pi)}$, over $d$-dimensions is
\begin{align}
   \mathcal{J} =\sqrt{\displaystyle\prod_{k=1}^d g_{kk} }.
\end{align}
%% Has been verified: see script for algebra: /Users/gjohnson/Desktop/Gkeyll_Install/Paper on Hamiltonian Flows/Extra_Cas/det_norm_coords.m

\subsection{Verification of the Momentum-Normalized Bracket}

A straightforward verification of the rather complicated normalized bracket is to compare the resulting equations of motion directly with the geodesic equation. Using equation \ref{eqn:norm_p_x_dot} to eliminate $\hat{p}_j$ in equation \ref{eqn:norm_p_p_dot}, then differentiating equation \ref{eqn:norm_p_x_dot} with respect to time and eliminating $\dot{\hat{p}}_k$, as well as using the shift identity from Section \ref{sec:Hamiltonian_Formulations}, we have a second-order equation for the motion given by
\begin{align}
    \ddot{x}^k
    = 
    \sum_{i,j,l}
    \frac{ \eta^{kj} \vnorm{\basis{l}}  }{  \vnorm{\basis{k}} }
    \bigg(
      -\pfrac{\eta_{jl}}{x^i} 
      -\frac{\eta_{jl}}{\vnorm{\basis{k}}} \pfrac{\vnorm{\basis{k}}}{x^i} 
      - \frac{\eta_{jl}}{\vnorm{\basis{j}}} \pfrac{\vnorm{\basis{j}}}{x^i}
      + \frac{ \eta_{il}}{\vnorm{\basis{j}}} \pfrac{\vnorm{\basis{i}}}{x^j}
      + \frac{1}{2}\frac{ \vnorm{\basis{i}} }{ \vnorm{\basis{j}} }\bigg( \pfrac{\eta_{il}}{x^j} \bigg)
    \bigg) \dot{x}^i \dot{x}^l
    \label{eqn:norm_p_x_dot_dot}
\end{align}

The second-order equation of motion arising from the non-canonical bracket can be compared to the geodesic equation,
\begin{align}
    \ddot{x}^k &= - \sum_{i,l} \Gamma\indices{^k_{il}} \dot{x}^i \dot{x}^l,\label{eqn:geodesic_with_explicit_sum}
\end{align}
by expanding the Christoffel symbol on the right-hand side and performing a change of variables to normalized momentum coordinates. Rewriting the Christoffel symbol in equation \ref{eqn:geodesic_with_explicit_sum} of the metric coefficients, we have 
\begin{align}
    \ddot{x}^k &= - \sum_{i,j,l} g^{jk} \left( \pfrac{g_{ji}}{x^l} - \frac{1}{2} \pfrac{g_{il}}{x^j} \right) \dot{x}^i \dot{x}^l.
\end{align}
Then, replace the metric components of $g$ in terms of $\eta$ from equations \ref{eqn:eta_ij} and \ref{eqn:eta_ij_inv}, and simplify the expression. 
\begin{align}
    \ddot{x}^k 
    &= 
    - \sum_{i,j,l}  
    \frac{\eta^{jk}}{\vnorm{\basis{j}}\vnorm{\basis{k}}}
    \bigg( \pfrac{}{x^l} 
        \left( 
        \eta_{ji} \vnorm{\basis{j}} \vnorm{\basis{i}} 
        \right) 
    - \frac{1}{2}\pfrac{}{x^j} 
        \left( 
        \eta_{il} \vnorm{\basis{i}} \vnorm{\basis{l}}
        \right)
    \bigg)  \dot{x}^i \dot{x}^l \\
    \ddot{x}^k 
    &= 
    - \sum_{i,j,l}  
    \frac{\eta^{kj}}{\vnorm{\basis{j}}\vnorm{\basis{k}}}
    \bigg( 
        \eta_{ji} \vnorm{\basis{j}} \pfrac{\vnorm{\basis{i}} }{x^l} 
        +
        \eta_{ji} \pfrac{ \vnorm{\basis{j}} }{x^l}  \vnorm{\basis{i}}
        +
         \pfrac{ \eta_{ji} }{x^l}  \vnorm{\basis{j}} \vnorm{\basis{i}} \nonumber \\
    &\quad  - \frac{1}{2} 
        \bigg( 
            \underbrace{
                \eta_{il} \vnorm{\basis{i}} \pfrac{\vnorm{\basis{l}}}{x^j} 
                +
                \eta_{il} \pfrac{ \vnorm{\basis{i}}}{x^j} \vnorm{\basis{l}}
            }_{
              \eta_{il} \vnorm{\basis{i}} \pfrac{\vnorm{\basis{l}}}{x^j}\dot{x}^i \dot{x}^l 
              =
              \eta_{il} \vnorm{\basis{l}} \pfrac{\vnorm{\basis{i}}}{x^j}\dot{x}^i \dot{x}^l 
            }
            +
            \pfrac{\eta_{il}}{x^j} \vnorm{\basis{i}} \vnorm{\basis{l}}
        \bigg)
    \bigg) \dot{x}^i \dot{x}^l \\
    \ddot{x}^k
    &= 
    \sum_{i,j,l}
    \frac{ \eta^{kj} \vnorm{\basis{l}}  }{  \vnorm{\basis{k}} }
    \bigg(
      -\pfrac{\eta_{jl}}{x^i} 
      -\frac{\eta_{jl}}{\vnorm{\basis{k}}} \pfrac{\vnorm{\basis{k}}}{x^i} 
      - \frac{\eta_{jl}}{\vnorm{\basis{j}}} \pfrac{\vnorm{\basis{j}}}{x^i}
      + \frac{ \eta_{il}}{\vnorm{\basis{j}}} \pfrac{\vnorm{\basis{i}}}{x^j}
      + \frac{1}{2}\frac{ \vnorm{\basis{i}} }{ \vnorm{\basis{j}} }\bigg( \pfrac{\eta_{il}}{x^j} \bigg)
    \bigg) \dot{x}^i \dot{x}^l
    \label{eqn:norm_p_x_dot_dot_from_geodesic}
\end{align}
Equation \ref{eqn:norm_p_x_dot_dot_from_geodesic} exactly matches equation \ref{eqn:norm_p_x_dot_dot}, verifying that the introduction of the normalized momentum coordinates has not changed the equations of motion that the particles follow.

\subsection{Example: Normalized Polar Coordinates}

Now, revisiting the annular disk scenario from the beginning of this section, let's apply the momentum coordinate normalization. Consider an extended domain of $r = 0.05$ to $1.5$ and $\theta = 0$ to $2\pi$ with a distribution at constant temperature, $T = 1$, no initial flows, and uniform density. In canonical momentum components, the momentum bounds needed to resolve the temperature at the inner and outer radial boundaries are quite large; the radial dependence introduces a factor of $r_{max}/r_{min} = 30$ in the required momentum range to resolve the innermost point for equivalent resolution at the outermost point.

In general, an estimate for the momentum domain bounds required for this problem can be approximated by identifying where the momentum in each coordinate direction has reduced the distribution by a thermal width. We can then multiply this width by a factor $C$ to obtain
\begin{align}
    p_{i,max} \sim C\sqrt{2T/g^{ii}}, \label{eqn:approx_bound_norm}
\end{align}
where a value of $C = 4$ corresponds to four thermal widths. For orthogonal coordinate systems, such as the ones we have primarily discussed so far — flat, annulus, spherical, and hyperbolic — equation \ref{eqn:approx_bound_norm} is exactly the momentum distance from the origin at which the distribution has dropped by $f_{lte}/A = \exp(-C)$, assuming zero bulk momentum. For these orthogonal but not orthonormal coordinate systems, the normalized non-canonical bracket is most effective at reducing the grid size. Shear caused by non-orthogonal bases complicates the matter, and the normalized momentum bracket is not guaranteed to reduce the grid requirements to resolve a given temperature.

Figure \ref{fig:annulus_norm_grid} shows an example of grid cells for the annulus. The slices in both canonical coordinates and normalized coordinates visually demonstrate the dependence of momentum cell width on radius in the canonical case. For instance, the distribution at the outermost radial point (panel (b) at $r = 1.5$) can be resolved with $8-12$ cells in $p_\theta$. The innermost point at $r = 0.05$ requires $30$ times as many cells in $p_\theta$ to resolve the same temperature with the same number of cells in $p_\theta$. The normalized momentum coordinates, $\hat{p}_\theta$, in panel (c) represent the same distribution function, but because of the removal of the geometric factor, remain uniform in momentum space at any point on the grid.

We emphasize the trade-off that is being outlined here with the normalized momentum coordinate approach: we proved $L_2$ stability of the canonical bracket in Section \ref{sec:scheme}, but found that $L_2$ stability could not be shown for a non-canonical bracket. Nevertheless, the increased efficiency of the phase space representation in these normalized coordinates is a vitally important option for large-scale calculations in certain coordinate systems, such as the polar coordinates exemplified here, spherical coordinates, or generalized to spherical Boyer-Lindquist or Kerr-Schild coordinates in general relativity. These geometries all have factors of radius that influence the canonical momentum space associated with the angular coordinates. Therefore, these geometries would all benefit from this approach by reducing the computational and memory costs of representing the problem discretely. The flexibility of our approach is thus made manifest: we can discretize both canonical and non-canonical brackets and leverage these trade-offs for a variety of different coordinate systems and physical systems. In fact, we envision that the generalization of this approach, using locally orthonormal bases such as triads or tetrads to further optimize our momentum space representation in the presence of e.g., shear, will naturally connect the DG foundation laid out here with general relativity to enable efficient large-scale calculations of systems such as compact objects.

% Annular disk figure (Uniform T and N, divide out J)
\begin{figure}[H]
\centering
\includegraphics[width = 1.0\linewidth]{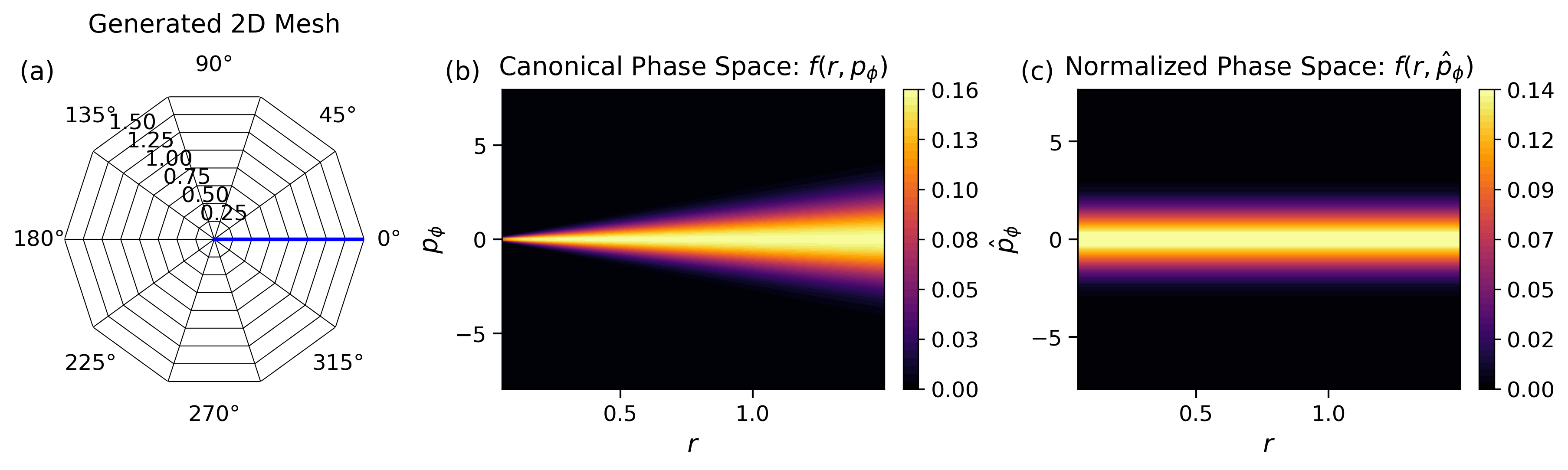}
\caption{ Phase space of a constant density $n = 1$ and temperature $T = 1$. Panel (a) is a 2D mesh of an annular grid of size $Nr \times N\theta  = 8 \times 10$. The blue line indicates the slices along which the phase space is plotted in the other two panels. Panel (b) is the phase space in canonical coordinates $r$ and $p_\theta$. Panel (c) plots the distribution function in normalized momentum coordinates $\hat{p}_\theta$. }\label{fig:annulus_norm_grid}
\end{figure}
\section{Conclusions and a Prospectus for Future Work} \label{sec:Conclusions}
In this paper, we have constructed a conservative discontinuous Galerkin scheme for the kinetic equation on curved manifolds. Free-streaming motion is formulated in terms of the Hamiltonian, whose structure encodes the geometry, and whose partial derivatives define the corresponding kinetic equation. In addition to our canonical scheme conserving the total number of particles, by enforcing the continuity of the discretely represented Hamiltonian, the normal characteristics are continuous, which allows us to prove that the discrete scheme is conservative for the total energy. For the non-canonical discrete system, conservation laws are maintained when the discrete Poisson tensor and Jacobian are also represented continuously. 

Collisions on curved surfaces are incorporated through a moment-preserving BGK operator. This operator handles collisions implicitly, without a time-step restriction, while retaining the conservation laws of the overall scheme. In the fluid (high-collisionality) limit, the scheme reproduces Euler's equations on curved manifolds and converges to the exact Riemann solution, demonstrating that the discrete scheme is asymptotic preserving.

A key result of this formulation is that the continuum scheme can, for arbitrary geometries, compute noise-free and alias-free solutions to the kinetic equation. This formulation thus enables highly accurate computations of small perturbations compared to the local equilibrium, from which one can compute transport coefficients and closures for kinetic corrections to reduced models. The utility of this formulation is highlighted for problems such as the Kelvin-Helmholtz instability, where detailed phase-space structure can be extracted from the kinetic distribution and local equilibrium. Furthermore, we have shown extensions of the Hamiltonian system, such as uniform rotation of a sphere. Additional physics can therefore be added to the solver by modifying the Hamiltonian to account for effects such as rotation while maintaining the conservation and stability properties.

One component of the success of the canonical scheme is the proof and demonstration of the $L_2$ stability of the discrete scheme; however, canonical coordinates are not always optimal. Canonical coordinates depend on local curvature, which sometimes puts stringent requirements on the needed momentum-space resolution. We show that this geometric dependence can be completely eliminated for orthogonal coordinate systems by using a non-canonical bracket, which maintains discrete conservation laws and the same underlying equations of motion, However, in the non-canonical case one does not inherit the same $L_2$ stability due to the lack of discrete phase-space incompressibility. Nevertheless, the improved efficiency of the phase space representation with non-canonical coordinates indicates a future direction of research into a generalized formulation via triads or tetrads (Cartan's methods of mobiles) in which the momentum-space basis is orthonormal. This triad/tetrad approach would naturally lead to relativistic extensions, where tetrad frames can describe local physics, such as collisions, while the distribution function flows along geodesics in the presence of gravity.

We conclude by emphasizing the success of the modeling approach outlined here. A Hamiltonian formulation using canonical or non-canonical coordinates can model a diverse array of systems, and the successful implementation of our approach is demonstrated on a collection of neutral gas test cases spanning the collisionless to collisional limit in non-trivial geometries. Indeed, the formalism derived, implemented, and tested in this manuscript provides a natural extension to previous modeling efforts coupling neutral particle dynamics to magnetized plasmas in fusion reactor geometries \cite{bernard2022kinetic,bernard2023blob}. The flexibility of this approach, starting with a Hamiltonian and making judicious choices of coordinates for the problem of interest, combined with developments in discontinuous Galerkin schemes for kinetic equations, provides a powerful platform for solving the Boltzmann equation. 

Future work will extend this Hamiltonian formulation to general relativity, where gravity manifests as spacetime curvature. The discrete scheme in canonical coordinates for free-streaming particles in curved spacetime is structurally identical and inherits the same properties when applied to the corresponding general relativistic Hamiltonian. However, a reformulation is required to incorporate processes such as collisions via a tetrad representation, which we will present in a future publication. The resulting relativistic solver will provide a new modeling foundation for kinetic physics, as well as provide the ability to compute closures and transport coefficients around compact objects in mixed collisional regimes. We envision this solver finding general use in modeling systems such as current sheets, jets around compact objects, and high-density-gradient regions such as pulsar magnetospheres.
\appendix

\section{Useful Identities for Multivariate Gaussian Distributions}

Consider the multivariate Gaussian distribution
\begin{align}
  f = \exp\left(-\frac{1}{2} \mvec{p}^ T\mvec{A} \mvec{p}\right)
\end{align}
where $\mvec{A}$ is a symmetric positive-definite $d\times d$ matrix
and $\mvec{p}$ is a vector. We would like to compute the moments of this
function. Consider, for example,
\begin{align}
  I_0 = \int_{-\infty}^\infty f \thinspace d\mvec{p}.
\end{align}
As $\mvec{A}$ is a symmetric positive-definite matrix, we can decompose
it as
\begin{align}
  \mvec{A} = \mvec{R}\gvec{\Lambda}\mvec{R}^{T}
\end{align}
where $\mvec{R}$ is a matrix with right-eigenvectors as columns with
$\mvec{R} \mvec{R}^{T} = I$, and
$\gvec{\Lambda} = \mathrm{diag}(\lambda_1,\ldots,\lambda_d)$ is a
diagonal matrix with eigenvalues $\lambda_i$ on the diagonal. Furthermore, as we have $\mvec{R}^{-1} = \mvec{R}^T$,  $\mvec{R}$ is orthogonal, and we can choose the eigenvector signs such that $\det\mvec{R} = \det\mvec{R}^{T} = 1$.

Now, define $\mvec{y} = \mvec{R}^{T}\mvec{p}$. Then we can write the
multivariate Gaussian as
\begin{align}
  f = \prod_{i=1}^d \exp\left(-\frac{1}{2} y_i^2 \lambda_i\right)
\end{align}
Further, we have $d\mvec{p} = \det\mvec{R}\thinspace d\mvec{y} =
d\mvec{y}$. Hence, we can compute
\begin{align}
  I_0 = \prod_{i=1}^d \int_{-\infty}^\infty \exp\left(-\frac{1}{2} y_i^2
  \lambda_i\right)\thinspace dy_i
  =
  \prod_{i=1}^d \sqrt{\frac{2\pi}{\lambda_i}} =
  \frac{(\sqrt{2\pi})^d}{\sqrt{\det \mvec{A}}}.
  \label{eq:I0-multivariate-gauss}
\end{align}
Here we used the fact that the product of the eigenvalues is the
determinant\footnote{We have the identities, for $a>0$,
  \begin{align}
    \int_{-\infty}^\infty e^{-y^2 a/2} dy &= \sqrt{\frac{2\pi}{a}} \\
    \int_{-\infty}^\infty y^2 e^{-y^2 a/2} dy &= \sqrt{\frac{2\pi}{a^3}}.
  \end{align}
}.

Next, we will compute the integral
\begin{align}
  \mvec{I}_2 =
  \int_{-\infty}^\infty \mvec{p}\otimes \mvec{p} f \thinspace d\mvec{p}.
\end{align}
We can write this in component form as
\begin{align}
  I_{ij} 
  =
  \int_{-\infty}^\infty p_i p_j f \thinspace d\mvec{p}
  = 
  R_{im} R_{jn} \int_{-\infty}^\infty y_m y_n 
  \prod_{k=1}^d \exp\left(-\frac{1}{2} y_k^2 \lambda_k\right) 
  \thinspace dy_k.
\end{align}
As the odd moments of $f$ vanish, this integral is non-zero only if
$m=n$. Hence, we get (see the identity in the footnote)
\begin{align}
  I_{ij} 
  =
  R_{im} R_{jn} \frac{\delta_{mn}}{\lambda_m} \prod_{k=1}^d \sqrt{\frac{2\pi}{\lambda_k}}
  =
  R_{im} R_{jn} \frac{\delta_{mn}}{\lambda_m} \frac{(\sqrt{2\pi})^d}{\sqrt{\det \mvec{A}}}.
\end{align}
Now, we can write
\begin{align}
  R_{im} \frac{\delta_{mn}}{\lambda_m} R_{jn} = 
  \mvec{R} \gvec{\Lambda}^{-1} \mvec{R}^T = \mvec{A}^{-1}.
\end{align}
Hence, we obtain the final expression
\begin{align}
  \mvec{I}_2 = \mvec{A}^{-1} \frac{(\sqrt{2\pi})^d}{\sqrt{\det \mvec{A}}}.
  \label{eq:I2-multivariate-gauss}
\end{align}
Often, we need to compute instead
\begin{align}
  I_2 = \int_{-\infty}^\infty \mvec{p}\cdot\mvec{p} f \thinspace d\mvec{p} 
  =
  \Tr\mvec{I}_2 
  = 
  \Tr \mvec{A}^{-1} \frac{(\sqrt{2\pi})^d}{\sqrt{\det \mvec{A}}}.
\end{align}

\section{Some Weak Identities for Manipulating Discrete Weak Forms} \label{app:weak_ids}

Properties of the discrete scheme were proved by manipulating the discrete weak form. This is a tricky process, even with care. Invalid manipulations are hard to identify and easy to commit without a set of allowable identities. Here, we introduce some useful weak identities as well as invalid manipulations that are useful for proving the properties of the discrete scheme.

Central to the weak identities is the space in which the functions reside, $\mathcal{V}_h^p$,
\begin{equation}
    \mathcal{V}_h^p = \{ w:w|_{K_j} \in \mvec{P}^p,\forall K_j \in \mathcal{T} \}.
\end{equation}
which is a piecewise discontinuous polynomial space. This space contains polynomial functions $w(\mvec{z})$ in each cell $K_j$ and belongs to the space of linearly combined basis polynomials $\mvec{P}^p$. $\mvec{z}$ is the phase space within each cell. To clearly label quantities in this space, we label them with a subscripted $h$. Projections can include individual functions, such as $f_h$ or groupings of functions, such as $(fg)_h$. The general identity from which everything is derived is that for an arbitrary function $f$, we have 
\begin{equation}
    \int_{K_j} w_h f d\mvec{z} = \int_{K_j} w_h f_h d\mvec{z}
\end{equation} 
which essentially shows that the integral times the test function $w$ is a projection operator that takes arbitrary functions and projects them onto $\mathcal{V}_h^p$. From this, we have a few specific identities that serve as algebraic steps within the proofs for discrete properties.

The first weak identity
\begin{equation}
    \int_{K_j} w_h f_h g_h d\mvec{z} = \int_{K_j} w_h (f g)_h d\mvec{z}
\end{equation} 
for functions $w$, $f_h$, and $g_h$ in $\mathcal{V}_h^p$.

The second weak identity prevents arbitrary grouping 
\begin{equation}
    \int_{K_j} w_h (f g)_h h_h d\mvec{z} \neq \int_{K_j} w_h f_h g_h h_h d\mvec{z} 
\end{equation} 
for functions $w$, $f_h$, $g_h$, and $h_h$ in $\mathcal{V}_h^p$ 

The third weak identity is a generalization of the first to an arbitrary number of functions, $m$. For $f^n$ belonging to $\mathcal{V}_h^p$ for $n = 1 ... m$
\begin{equation}
    \int_{K_j} w_h \left( \prod_{n=1}^{m} f^n_h \right) d\mvec{z} = \int_{K_j} w_h \left( \prod_{n=1}^{m} f^n \right)_h d\mvec{z}
\end{equation} 

% This has been verified (using the generalized weak identity)
Lastly, we have a derivative identity. If $f_h$ satisfies
\begin{equation}
    \frac{df_h}{dx}  \doteq 0
\end{equation}
then we may write
\begin{equation}
    \int_{K_j} \frac{d}{dx} \left( f_h g_h \right) d\mvec{z} 
    =
    \int_{K_j} \frac{df_h}{dx} g_h d\mvec{z} 
\end{equation} 
where product rules and other vector identities may be applied within the cell where discrete functions are polynomial and are hence smooth.

\section{Geodesics on the Surface of a Sphere and Hyperbolic Surface}\label{app:geodesics}
This appendix elaborates on geodesic motion on the surface of a sphere and in hyperbolic geometry, which underlie the Kelvin-Helmholtz examples in Section \ref{sec:Benchmarks}. To avoid coordinate singularities at the poles of the sphere, we reflect particles at latitudes of $\theta_{min} = \pi/4$ and $\theta_{max} =3\pi/4)$. For the hyperbolic surface, reflections occur at $z = \pm 1$. These reflections are implemented to mirror how the distribution function is affected by changing the sign of the momentum component normal to the boundary. To visualize the motion in these geometries, Figure \ref{fig:geodesic_plots} plots representative geodesics.

For the surface of a sphere, we have the line element,
\begin{equation}
    ds_{sphere}^2 = r^2 d\theta^2 + r^2 \sin^2\theta d\phi^2,
\end{equation}
where the metric can be extracted immediately via $ds^2 = g_{ij}dx^i dx^j$. Using this metric in the Hamiltonian $H = \frac{1}{2}g^{ij}p_ip_j$ and computing Hamilton's equations, the resulting equations of motion for the sphere can be combined into two coupled second-order ODEs
\begin{equation}\label{eqn:surf_sphere_geodesic_ddot_theta}
    \ddot{\theta} = \frac{\dot{\phi}^2 \sin 2\theta}{2}
\end{equation}
\begin{equation}\label{eqn:surf_sphere_geodesic_ddot_phi}
   \ddot{\phi} = -\frac{2 \dot{\phi} \dot{\theta} \cos \theta}{\sin \theta}.
\end{equation}
These are the equations for great-circle geodesics. Therefore, any particle (not interacting with a reflecting $\theta$-boundary) will trace out a great circle shown in Panel (a) of Figure \ref{fig:geodesic_plots}.  If the geodesic of the particle is such that its path intersects the reflecting boundary, then the free streaming trajectory is significantly altered. Figure \ref{fig:geodesic_plots} panel (b) shows such a case. The particle reflects several times as it traverses the azimuthal angle of the sphere, eventually returning near its initial position but with an altered phase.

For a hyperbolic surface of constant $r$, the line element is
\begin{equation}
    ds_{hyperbolic}^2 = (r^2 +z^2) d\theta^2 + \frac{r^2 + 2z^2}{r^2 + z^2} dz^2.
\end{equation}
By computing the equations of motion for this metric in the free-streaming Hamiltonian, we obtain 
\begin{equation}\label{eqn:surf_hyperbolic_geodesic_ddot_theta}
    \ddot{\theta} = -\frac{2 \dot{\theta} z \dot{z}}{r^2 + z^2}
\end{equation}
\begin{equation}\label{eqn:surf_hyperbolic_geodesic_ddot_z}
    \ddot{z} = \frac{\dot{\theta}^2 z (r^2 + z^2)}{r^2 + 2z^2} - \frac{r^2 z \dot{z}^2}{r^4 + 3r^2 z^2 + 2z^4}.
\end{equation}
The geodesics in hyperbolic geometry, like those in spherical geometry, conserve angular momentum, $p_\theta$, due to the metric independence of $\theta$ in both cases. In the example trajectories of panels (c, d) in Figure \ref{fig:geodesic_plots}, the particles drift away from $z = 0$ until they interact with the boundaries at $z = \pm 1$ and reflect back toward $z = 0$.

This appendix has highlighted the trajectories that particles will take and how interaction with the boundaries of the domain affects their motion. For the simulation cases using these geometries in Section \ref{sec:Benchmarks}, the boundaries do not influence the initial growth of the Kelvin-Helmholtz instability, which occurs away from the boundaries as the shear layer is initialized in the interior of the domain. However, these boundary reflections must be included to reproduce the non-linear saturated state when the eddies begin interacting with the edges of the domain.

\begin{figure}[H]
\centering
\includegraphics[width = 1.0\linewidth]{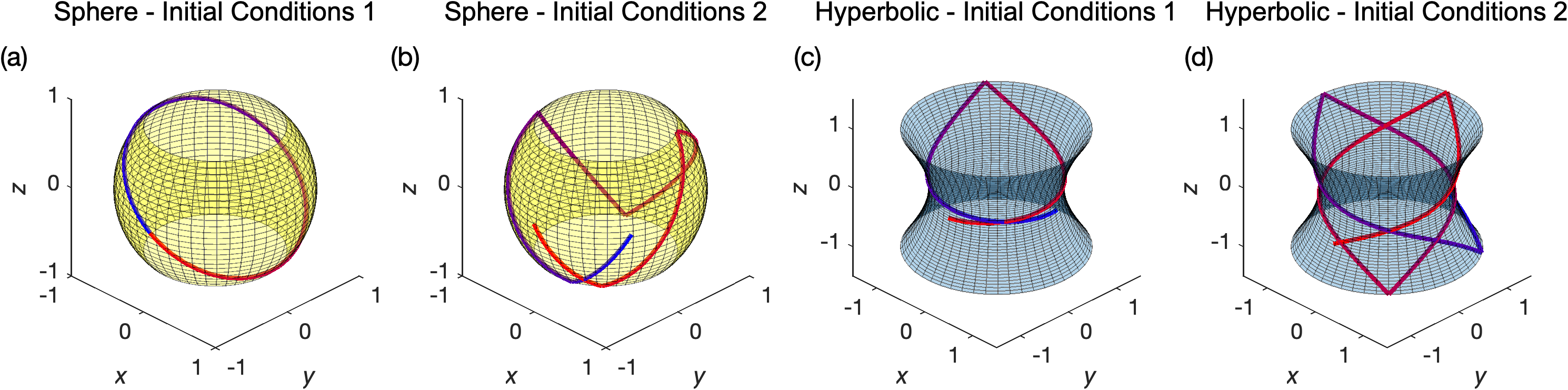}
\caption{Sample geodesics on the surfaces of a sphere (left) and hyperbolic surface (right). The color of the geodesic lines corresponds to initial time (red, $t_0 = 0$) and final time (blue, $t_f$). Two sets of initial conditions were used to demonstrate different paths. Initial condition 1 for the sphere of radius $r = 1$ is: $[\theta_0, \phi_0, p_{\theta,0}, p_{\phi,0}, t_f] = [\pi/2, -\pi/2, 0.8, 1.0, 5]$ for panel (a) and $[\pi/2, -\pi/2, 1.5, 1.0, 5.65]$ for panel (b). For the hyperbolic initial conditions with a fixed radial parameter $r = 0.5$: $[\theta_0, z_0, p_{\theta,0}, p_{z,0}, t_f] = [-\pi/2, -0.1, 1.5, 0.2, 6]$ for panel (c), and $[-\pi/2, -0.5, 1.5, 0.9, 10]$ for panel (d). }  \label{fig:geodesic_plots}
\end{figure}
\section*{Funding}
This work was supported by the U.S. Department of Energy under contract number DE-AC02-09CH11466. G.J. was supported by the U.S. Department of Energy, Office of Science, Office of Advanced Scientific Computing Research, Department of Energy Computational Science Graduate Fellowship under Award Number DE-SC0021110 and by the U.S. Department of Energy under Contract No. DE-AC02-09CH11466 via LDRD grants. A.H. and J.J. were supported by the U.S. Department of Energy under Contract No. DE-AC02-09CH11466 via LDRD grants. This work is supported by a DOE Distinguished Scientist award, the CEDA SciDAC project, and other PPPL projects via DOE Contract No. DE-AC02-09CH11466 for the Princeton Plasma Physics Laboratory.

\bibliographystyle{elsarticle-num}  % or elsarticle-num, elsarticle-num-names
\bibliography{references}

@article{bgk1954,
  title = {A Model for Collision Processes in Gases. I. Small Amplitude Processes in Charged and Neutral One-Component Systems},
  author = {Bhatnagar, P. L. and Gross, E. P. and Krook, M.},
  journal = {Phys. Rev.},
  volume = {94},
  issue = {3},
  pages = {511--525},
  numpages = {0},
  year = {1954},
  month = {May},
  publisher = {American Physical Society}
}

@article{hakim2019discontinuous,
  title={{Discontinuous Galerkin Schemes for a Class of Hamiltonian Evolution Equations with Applications to Plasma Fluid and Kinetic Problems}},
  author={Hakim, A and Hammett, G and Shi, E and Mandell, N},
  journal={arXiv preprint arXiv:1908.01814},
  year={2019}
}

@article{arnold2011serendipity,
  title={{The Serendipity Family of Finite Elements}},
  author={Arnold, Douglas N and Awanou, Gerard},
  journal={Foundations of Computational Mathematics},
  volume={11},
  number={3},
  pages={337--344},
  year={2011},
  publisher={Springer}
}

@inproceedings{hakim2020alias,
  title={{Alias-free, Matrix-free, and Quadrature-free Discontinuous Galerkin Algorithms for (Plasma) Kinetic Equations}},
  author={Hakim, Ammar and Juno, James},
  booktitle={SC20: International Conference for High Performance Computing, Networking, Storage and Analysis},
  pages={1--15},
  year={2020},
  organization={IEEE}
}

@article{shu1988efficient,
  title={{Efficient Implementation of Essentially Non-Oscillatory Shock-Capturing Schemes}},
  author={Shu, Chi-Wang and Osher, Stanley},
  journal={Journal of computational physics},
  volume={77},
  number={2},
  pages={439--471},
  year={1988},
  publisher={Elsevier}
}

@article{cary2009hamiltonian,
  title={{Hamiltonian Theory of Guiding-center Motion}},
  author={Cary, John R and Brizard, Alain J},
  journal={Reviews of modern physics},
  volume={81},
  number={2},
  pages={693--738},
  year={2009},
  publisher={APS}
}

@article{mandell2020electromagnetic,
  title={{Electromagnetic Full-Gyrokinetics in the Tokamak Edge with Discontinuous Galerkin Methods}},
  author={Mandell, NR and Hakim, A and Hammett, GW and Francisquez, M},
  journal={Journal of Plasma Physics},
  volume={86},
  number={1},
  pages={905860109},
  year={2020},
  publisher={Cambridge University Press}
}

@book{shi2017gyrokinetic,
  title={{Gyrokinetic Continuum Simulation of Turbulence in Open-Field-Line Plasmas}},
  author={Shi, Eric Leon},
  year={2017},
  publisher={Princeton University}
}

@article{font1998grhydro,
    author = {Font, José A. and Ibáñez, J. M.},
    title = {Non-axisymmetric relativistic Bondi—Hoyle accretion on to a Schwarzschild black hole},
    journal = {Monthly Notices of the Royal Astronomical Society},
    volume = {298},
    number = {3},
    pages = {835-846},
    year = {1998},
    month = {08},
    issn = {0035-8711}
}

@article{font1999kerr,
    author = {Font, José A. and Ibáñez, José Maria and Papadopoulos, Philippos},
    title = {Non-axisymmetric relativistic Bondi-Hoyle accretion on to a Kerr black hole},
    journal = {Monthly Notices of the Royal Astronomical Society},
    volume = {305},
    number = {4},
    pages = {920-936},
    year = {1999},
    month = {05},
    issn = {0035-8711}
}

@article{anton2006grmhd,
year = {2006},
month = {jan},
publisher = {},
volume = {637},
number = {1},
pages = {296},
author = {Antón, Luis and Zanotti, Olindo and Miralles, Juan A. and Martí, José M. and Ibáñez, José M. and Font, José A. and Pons, José A.},
title = {Numerical 3+1 General Relativistic Magnetohydrodynamics: A Local Characteristic Approach},
journal = {The Astrophysical Journal},
}

@article{galishnikova2025entity,
  title={Entity--Hardware-agnostic Particle-in-Cell Code for Plasma Astrophysics. II: General Relativistic Module},
  author={Galishnikova, Alisa and Hakobyan, Hayk and Philippov, Alexander and Crinquand, Benjamin},
  journal={arXiv preprint arXiv:2511.17701},
  year={2025}
}

@book{wald2010general,
  title={{General Relativity}},
  author={Wald, Robert M},
  year={2010},
  publisher={University of Chicago press}
}

@article{parfrey2019first,
  title={{First-Principles Plasma Simulations of Black-Hole Jet Launching}},
  author={Parfrey, Kyle and Philippov, Alexander and Cerutti, Beno{\^\i}t},
  journal={Physical review letters},
  volume={122},
  number={3},
  pages={035101},
  year={2019},
  publisher={APS}
}

@article{crinquand2020multidimensional,
  title={{Multidimensional Simulations of Ergospheric Pair Discharges around Black Holes}},
  author={Crinquand, Benjamin and Cerutti, Beno{\^\i}t and Philippov, Alexander and Parfrey, Kyle and Dubus, Guillaume},
  journal={Physical Review Letters},
  volume={124},
  number={14},
  pages={145101},
  year={2020},
  publisher={APS}
}

@article{galishnikova2023collisionless,
  title={{Collisionless Accretion onto Black Holes: Dynamics and Flares}},
  author={Galishnikova, Alisa and Philippov, Alexander and Quataert, Eliot and Bacchini, Fabio and Parfrey, Kyle and Ripperda, Bart},
  journal={Physical review letters},
  volume={130},
  number={11},
  pages={115201},
  year={2023},
  publisher={APS}
}

@article{komissarov2004general,
  title={{General Relativistic Magnetohydrodynamic Simulations of Monopole Magnetospheres of Black Holes}},
  author={Komissarov, SS},
  journal={Monthly Notices of the Royal Astronomical Society},
  volume={350},
  number={4},
  pages={1431--1436},
  year={2004},
  publisher={Blackwell Science Ltd Oxford, UK}
}

@article{komissarov2004electrodynamics,
  title={{Electrodynamics of Black Hole Magnetospheres}},
  author={Komissarov, SS},
  journal={Monthly Notices of the Royal Astronomical Society},
  volume={350},
  number={2},
  pages={427--448},
  year={2004},
  publisher={Blackwell Science Ltd Oxford, UK}
}

@article{dodin2010vlasov,
  title={{Vlasov Equation and Collisionless Hydrodynamics Adapted to Curved Spacetime}},
  author={Dodin, IY and Fisch, NJ},
  journal={Physics of Plasmas},
  volume={17},
  number={11},
  year={2010},
  publisher={AIP Publishing}
}

@article{cerutti2015particle,
  title={{Particle Acceleration in Axisymmetric Pulsar Current Sheets}},
  author={Cerutti, Beno{\^\i}t and Philippov, Alexander and Parfrey, Kyle and Spitkovsky, Anatoly},
  journal={Monthly Notices of the Royal Astronomical Society},
  volume={448},
  number={1},
  pages={606--619},
  year={2015},
  publisher={Oxford University Press}
}

@article{chen2014electrodynamics,
  title={{Electrodynamics of Axisymmetric Pulsar Magnetosphere with Electron--Positron Discharge: A Numerical Experiment}},
  author={Chen, Alexander Y and Beloborodov, Andrei M},
  journal={The Astrophysical Journal Letters},
  volume={795},
  number={1},
  pages={L22},
  year={2014},
  publisher={IOP Publishing}
}

@article{philippov2014ab,
  title={{Ab Initio Pulsar Magnetosphere: Three-Dimensional Particle-in-Cell Simulations of Axisymmetric Pulsars}},
  author={Philippov, Alexander A and Spitkovsky, Anatoly},
  journal={The Astrophysical Journal Letters},
  volume={785},
  number={2},
  pages={L33},
  year={2014},
  publisher={IOP Publishing}
}

@article{witten2020light,
  title={Light rays, singularities, and all that},
  author={Witten, Edward},
  journal={Reviews of Modern Physics},
  volume={92},
  number={4},
  pages={045004},
  year={2020},
  publisher={APS}
}

@article{liu2025asdex,
    author = {Liu, D. and Juno, J. and Hammett, G. W. and Hakim, A. and Shukla, A. and Francisquez, M.},
    title = {Axisymmetric gyrokinetic simulation of ASDEX-Upgrade scrape-off layer using a conservative implicit BGK collision operator},
    journal = {Physics of Plasmas},
    volume = {32},
    number = {11},
    pages = {113906},
    year = {2025},
    month = {11},
    issn = {1070-664X},
}

@article{bernard2022kinetic,
  title={{Kinetic Modeling of Neutral Transport for a Continuum Gyrokinetic Code}},
  author={Bernard, TN and Halpern, FD and Francisquez, M and Mandell, NR and Juno, J and Hammett, GW and Hakim, A and Wilkie, GJ and Guterl, J},
  journal={Physics of Plasmas},
  volume={29},
  number={5},
  year={2022},
  publisher={AIP Publishing}
}

@article{bernard2023blob,
    author = {Bernard, T. N. and Halpern, F. D. and Francisquez, M. and Juno, J. and Mandell, N. R. and Hammett, G. W. and Hakim, A. and Humble, E. and Mukherjee, R.},
    title = {Effect of neutral interactions on parallel transport and blob dynamics in gyrokinetic scrape-off layer simulations},
    journal = {Physics of Plasmas},
    volume = {30},
    number = {11},
    pages = {112501},
    year = {2023},
    month = {11},
    issn = {1070-664X}
}

@article{morrison1980vlasov,
title = {The Maxwell-Vlasov equations as a continuous hamiltonian system},
journal = {Physics Letters A},
volume = {80},
number = {5},
pages = {383-386},
year = {1980},
issn = {0375-9601},
author = {Philip J. Morrison}
}

@article{marsden1982vlasov,
title = {The Hamiltonian structure of the Maxwell-Vlasov equations},
journal = {Physica D: Nonlinear Phenomena},
volume = {4},
number = {3},
pages = {394-406},
year = {1982},
issn = {0167-2789},
author = {Jerrold E. Marsden and Alan Weinstein},
}

@article{burby2017kinetic,
    author = {Burby, J. W.},
    title = {Finite-dimensional collisionless kinetic theory},
    journal = {Physics of Plasmas},
    volume = {24},
    number = {3},
    pages = {032101},
    year = {2017},
    month = {03},
    issn = {1070-664X}
}

@misc{goldstein2002classical,
  title={{Classical Mechanics}},
  author={Goldstein, Herbert and Poole, Charles and Safko, John},
  year={2002},
  publisher={American Association of Physics Teachers}
}

@article{juno2018discontinuous,
  title={{Discontinuous Galerkin Algorithms for Fully Kinetic Plasmas}},
  author={Juno, James and Hakim, Ammar and TenBarge, Jason and Shi, Eric and Dorland, William},
  journal={Journal of Computational Physics},
  volume={353},
  pages={110--147},
  year={2018},
  publisher={Elsevier}
}

@article{hakim2020conservative,
  title={{Conservative Discontinuous Galerkin Schemes for Nonlinear Dougherty--Fokker--Planck Collision Operators}},
  author={Hakim, Ammar and Francisquez, Manaure and Juno, James and Hammett, Gregory W},
  journal={Journal of Plasma Physics},
  volume={86},
  number={4},
  pages={905860403},
  year={2020},
  publisher={Cambridge University Press}
}

@article{dzanic2023positivity,
  title={{A Positivity-Preserving and Conservative High-Order Flux Reconstruction Method for the Polyatomic Boltzmann--BGK Equation}},
  author={Dzanic, Tarik and Witherden, Freddie D and Martinelli, Luigi},
  journal={Journal of Computational Physics},
  volume={486},
  pages={112146},
  year={2023},
  publisher={Elsevier}
}

@article{nevins2005discrete,
  title={Discrete particle noise in particle-in-cell simulations of plasma microturbulence},
  author={Nevins, WM and Hammett, Gregory W and Dimits, Andris M and Dorland, William and Shumaker, Dana E},
  journal={Physics of plasmas},
  volume={12},
  number={12},
  year={2005},
  publisher={AIP Publishing}
}

@article{juno2020noise,
  title={Noise-induced magnetic field saturation in kinetic simulations},
  author={Juno, J and Swisdak, MM and Tenbarge, JM and Skoutnev, V and Hakim, A},
  journal={Journal of Plasma Physics},
  volume={86},
  number={4},
  pages={175860401},
  year={2020},
  publisher={Cambridge University Press}
}

@article{chandra2015extended,
  title={An extended magnetohydrodynamics model for relativistic weakly collisional plasmas},
  author={Chandra, Mani and Gammie, Charles F and Foucart, Francois and Quataert, Eliot},
  journal={The Astrophysical Journal},
  volume={810},
  number={2},
  pages={162},
  year={2015},
  publisher={IOP Publishing}
}

@article{cordeiro2024causality,
  title={Causality bounds on dissipative general-relativistic magnetohydrodynamics},
  author={Cordeiro, Ian and Speranza, Enrico and Ingles, Kevin and Bemfica, F{\'a}bio S and Noronha, Jorge},
  journal={Physical review letters},
  volume={133},
  number={9},
  pages={091401},
  year={2024},
  publisher={APS}
}

@article{francisquez2020conservative,
  title={Conservative discontinuous Galerkin scheme of a gyro-averaged Dougherty collision operator},
  author={Francisquez, Manaure and Bernard, Tess N and Mandell, Noah R and Hammett, Gregory W and Hakim, Ammar},
  journal={Nuclear Fusion},
  volume={60},
  number={9},
  pages={096021},
  year={2020},
  publisher={IOP Publishing}
}

@article{johnson2025moment,
  title={A moment-conserving discontinuous Galerkin representation of the relativistic Maxwellian distribution},
  author={Johnson, Grant and Hakim, Ammar and Juno, James},
  journal={Journal of Plasma Physics},
  volume={91},
  number={5},
  pages={E130},
  year={2025},
  publisher={Cambridge University Press}
}

@article{pieraccini2007implicit,
  title={Implicit--explicit schemes for BGK kinetic equations},
  author={Pieraccini, Sandra and Puppo, Gabriella},
  journal={Journal of Scientific Computing},
  volume={32},
  number={1},
  pages={1--28},
  year={2007},
  publisher={Springer}
}

@article{pareschi2005implicit,
  title={Implicit--explicit Runge--Kutta schemes and applications to hyperbolic systems with relaxation},
  author={Pareschi, Lorenzo and Russo, Giovanni},
  journal={Journal of Scientific computing},
  volume={25},
  number={1},
  pages={129--155},
  year={2005},
  publisher={Springer}
}

@article{mieussens2000discrete,
  title={Discrete-velocity models and numerical schemes for the Boltzmann-BGK equation in plane and axisymmetric geometries},
  author={Mieussens, Luc},
  journal={Journal of Computational Physics},
  volume={162},
  number={2},
  pages={429--466},
  year={2000},
  publisher={Elsevier}
}

@article{shukla2025constructing,
  title={Constructing Field Aligned Coordinate Systems for Gyrokinetic Simulations of Tokamaks in X-point Geometries},
  author={Shukla, Akash and Hakim, Ammar and Juno, James and Hammett, Gregory and Francisquez, Manaure},
  journal={arXiv preprint arXiv:2510.21676},
  year={2025}
}

@book{arnold1989mathematical,
  title={Mathematical methods of classical mechanics},
  author={Arnold, Vladimir Igorevich and Vogtmann, Karen and Weinstein, Alan},
  volume={60},
  year={1989},
  publisher={Springer}
}

@book{durran2010numerical,
  title={Numerical methods for fluid dynamics: With applications to geophysics},
  author={Durran, Dale R},
  volume={32},
  year={2010},
  publisher={Springer Science \& Business Media}
}

@article{tadmor2025stability,
  title={On the stability of Runge--Kutta methods for arbitrarily large systems of ODEs},
  author={Tadmor, Eitan},
  journal={Communications on Pure and Applied Mathematics},
  volume={78},
  number={4},
  pages={821--855},
  year={2025},
  publisher={Wiley Online Library}
}

\end{document}